\documentclass[10pt,journal,compsoc]{IEEEtran-tkde}

\newtheorem{definition}{Definition}
\newtheorem{theorem}{Theorem}
\newenvironment{proof}{\begin{IEEEproof}}{\end{IEEEproof}}
\usepackage{multirow}  
\usepackage{amssymb}

\ifCLASSOPTIONcompsoc
  \usepackage[nocompress]{cite}
\else
  \usepackage{cite}
\fi
\ifCLASSINFOpdf
  \usepackage[pdftex]{graphicx}
  \DeclareGraphicsExtensions{.pdf,.jpeg,.png}
\else
  \usepackage[dvips]{graphicx}
  \DeclareGraphicsExtensions{.eps}
\fi
\usepackage{amsmath}
\usepackage{algorithm}
\usepackage{algorithmic}

\newcommand\MYhyperrefoptions{bookmarks=true,bookmarksnumbered=true,
pdfpagemode={UseOutlines},plainpages=false,pdfpagelabels=true,
colorlinks=true,linkcolor={blue},citecolor={blue},urlcolor={blue},
pdftitle={Fast Utility Mining on Complex Sequences},
pdfsubject={Typesetting},
pdfauthor={Wensheng Gan},
pdfkeywords={data mining, sequence, high-utility sequential pattern, efficiency, UL-list structure}}
\ifCLASSINFOpdf
\usepackage[\MYhyperrefoptions,pdftex]{hyperref}
\else
\usepackage[\MYhyperrefoptions,breaklinks=true,dvips]{hyperref}
\usepackage{breakurl}
\fi
\label{begin document}
\begin{document}
%
\title{Fast Utility Mining on Complex Sequences}
%
%
%
%


\author{Wensheng Gan,
	Jerry Chun-Wei Lin,~\IEEEmembership{Senior,~Member},
	Jiexiong Zhang,
	Philippe Fournier-Viger,\\
	Han-Chieh Chao,~\IEEEmembership{Senior,~Member}
	and Philip S. Yu,~\IEEEmembership{Fellow,~IEEE}
	\IEEEcompsocitemizethanks{
		\IEEEcompsocthanksitem Wensheng Gan is with Harbin Institute of Technology (Shenzhen), Shenzhen, China, and with University of Illinois at Chicago, IL, USA. Email: wsgan001@gmail.com
		
		\IEEEcompsocthanksitem Jerry Chun-Wei Lin is with the Western Norway University of Applied Sciences, Bergen, Norway. Email: jerrylin@ieee.org
		
		\IEEEcompsocthanksitem Jiexiong Zhang is with Harbin Institute of Technology (Shenzhen), Shenzhen, China. Email: jiexiong.zhang@foxmail.com
		
		\IEEEcompsocthanksitem Philippe Fournier-Viger is with Harbin Institute of Technology (Shenzhen), Shenzhen, China. Email: philfv8@yahoo.com
		
		\IEEEcompsocthanksitem Han-Chieh Chao is with the National Dong Hwa University, Hualien, Taiwan. Email: hcc@ndhu.edu.tw
		
		\IEEEcompsocthanksitem Philip S. Yu is with University of Illinois at Chicago, IL, USA. Email: psyu@uic.edu}
	
}

\IEEEtitleabstractindextext{%
\begin{abstract}
High-utility sequential pattern mining is an emerging topic in the field of Knowledge Discovery in Databases. It consists of discovering subsequences having a high utility (importance) in sequences, referred to as high-utility sequential patterns (HUSPs). HUSPs can be applied to many real-life applications, such as market basket analysis, E-commerce recommendation, click-stream analysis and scenic route planning. For example, in economics and targeted marketing, understanding economic behavior of consumers is quite challenging, such as finding credible and reliable information on product profitability. Several algorithms have been proposed to address this problem by efficiently mining utility-based useful sequential patterns. Nevertheless, the performance of these algorithms can be unsatisfying in terms of runtime and memory usage due to the  combinatorial explosion of the search space for low utility threshold and large databases. Hence, this paper proposes a more efficient algorithm for the task of high-utility sequential pattern mining, called HUSP-ULL. It utilizes a lexicographic sequence (LS)-tree and a utility-linked (UL)-list structure to fast discover HUSPs. Furthermore, two pruning strategies are introduced in HUSP-ULL to obtain tight  upper-bounds on the utility of candidate sequences, and reduce the search space by pruning unpromising candidates early. Substantial experiments both on real-life and synthetic datasets show that the proposed algorithm can effectively and efficiently discover the complete set of HUSPs and outperforms the state-of-the-art algorithms.
\end{abstract}

}

\maketitle

\section{Introduction}
Sequential pattern mining (SPM) \cite{agrawal1995mining,srikant1996mining,pei2001prefixspan,fournier2017survey} is an interesting and critical area of research in Knowledge Discovery in Databases (KDD) \cite{agrawal1993database,chen1996data}, which plays a key role in various applications such as DNA sequence analysis, consumer behavior analysis, and natural disaster analysis \cite{fournier2017survey}. The main objective of SPM is to discover a set of frequent sequences in a sequence database, selected with respect to a user-specified minimum support threshold, and where the frequency of each sequence is defined as its occurrence count in the database. Since data/information quality may be influenced by the sequential ordering of the events, for example, to assess the data/information quality on the Weblog data, it is thus important to take this attribute into account for providing more precise assessment of data/information quality.

SPM is similar to frequent itemset mining (FIM) \cite{agrawal1994fast,han2004mining}, as it is designed to discover patterns that frequently occur in data.  The implicit assumption of FIM and SPM is that frequent patterns are useful and interesting. For example, if may be found that numerous customers purchase beer and diapers together, which may be an interesting information for a business manager. The main difference between SPM and FIM is that SPM generalizes FIM by considering the sequential ordering of purchases. Thus, mining interesting patterns in a sequential database using SPM is  more challenging than FIM. 

One significant shortcoming of traditional sequential pattern mining is that all objects (items, event, sequence, movements, etc.) are treated equally. In fact, the most frequently occurring patterns can be, quite typically, the least interesting ones. In general, criteria such as the interestingness, weight, and importance of patterns are not taken into account in traditional SPM and FIM. Consequently, these frameworks can reveal many patterns that are frequent but uninteresting to decision makers. To better measure the importance of patterns to decision makers, other criteria can be considered such as the amount of profit (utility) that each pattern yields. For example, in market basket analysis, the item diamond may not be considered as a frequent pattern if its selling frequency is low, unlike the item egg. However, some infrequent patterns such as diamonds may yield a higher profit than egg. To address this issue, FIM was generalized to obtain the problem of high-utility itemset mining (HUIM) \cite{chan2003mining,yao2004foundational,2lin2016efficient,liu2012mining,liu2005two}. 
This latter takes both the purchase quantities and unit profits of items into account to identify the set of high-utility itemsets (HUIs). In HUIM, a quantitative value is associated to each item in each transaction to indicate the number of units of the item that were purchased. This is different from the traditional FIM problem where these values are restricted to binary values indicating that an item either appears or not in a transaction. Besides, an external profit value is associated to each item in HUIM to indicate its relative importance such as its unit profit or weight.  Because more HUIM takes more information into account than FIM and because the downward closure property of FIM (also called Apriori property \cite{agrawal1994fast}) does not hold in HUIM, HUIM is considered as more challenging than FIM.

\begin{figure}[!htbp]
	\setlength{\abovecaptionskip}{0pt}
	\setlength{\belowcaptionskip}{0pt} 		
	\centering
	\includegraphics[width=3.55in]{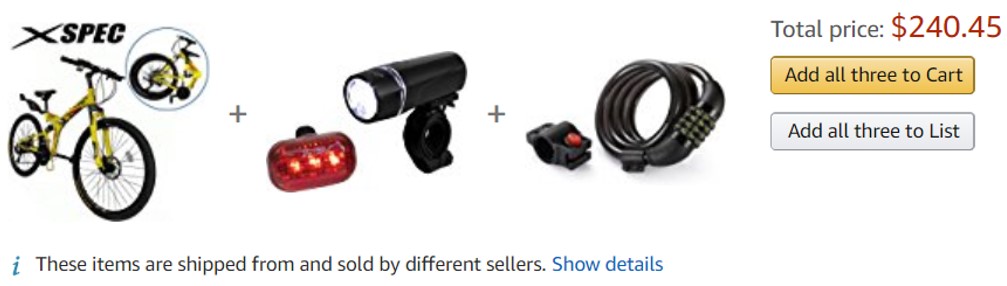}
	\caption{A shopping example in Amazon.}
	\label{fig:exampleOfHUSPM}
\end{figure}

Recently, to extract more informative patterns from  ordered data (sequences), SPM has been  generalized as the task of high-utility sequential pattern mining (HUSPM) \cite{yin2012uspan,yin2013efficiently,lan2014applying,wang2016efficiently}. Different from SPM, HUSPM considers not only the sequential ordering of items but also their utility values. Hence, HUSPM is more difficult than traditional SPM and HUIM. As shown in Fig. \ref{fig:exampleOfHUSPM}, a customer wants to purchase a mountain trail Bicycle, LED Headlight, and the UShake Bike lock, which have their unit prices. In general, these products/items are shipped from and sold by different sellers. Firstly, he/she will find out one low price mountain trail Bicycle from a seller, then continue to search the satisfied  LED Headlight with low price from another seller. Finally, he/she may buy a Bike lock from the same or different store. In this case, the consumers' purchase behavior consists of a series of utility-oriented sequential events/processes within different timestamps.

A sequence is said to be a high-utility sequential pattern (HUSP) if its total utility in a database is no less than a user-specified minimum utility threshold. HUSPs can be applied to many real-life applications, such as market basket analysis, E-commerce recommendation, click-stream analysis and scenic route planning. In HUSPM, however, the utility of a pattern is neither monotonic nor anti-monotonic. Therefore, the downward closure property of support (aka the Apriori property) does not hold in HUSPM and the search space is quite difficult to prune. Previous approaches have proposed methods for determining upper bounds on the utility of future candidate sequential patterns.

Since high-utility sequential pattern mining has many applications, more and more researchers are working on this problem. Several algorithms were respectively developed to efficiently mine the complete set of high-utility sequential patterns. However, these algorithms often consume a large amount of memory and have long execution times due to the combinatorial explosion of the search space. Besides, as information collection techniques are continuously  improved, the data that needs to be analyzed grows quickly. As a result, it is important to design more efficient algorithms to fast discover high-utility sequential patterns in large databases.

To address these problems, this paper designs a novel utility-linked list (UL-list) based algorithm called HUSP-ULL for mining the set of HUSPs more efficiently. The major contributions of this paper are as follows:

\begin{enumerate}
	\item \textbf{Insightful patterns}. A novel fast algorithm is proposed to efficiently identify meaningful and profitable HUSPs. It employs a utility-linked list structure and two pruning strategies to improve its mining performance.
	\item \textbf{Novel index structures}. A lexicographic-sequence (LS)-tree is introduced to represent the search space and mine the complete set of HUSPs. A compressed utility-linked (UL)-list structure is further designed to store information about patterns instead of processing the original database. UL-list is quite compact and different from the current existing data structures for utility mining.
	\item \textbf{Effective pruning}. In addition, two pruning strategies are integrated in the designed algorithm to reduce the search space and improve its performance to discover HUSPs.
	\item \textbf{Fast and better scalability}. Experimental results show that the proposed algorithm  can efficiently discover HUSPs and outperforms the existing state-of-the-art HUSPM algorithms, in terms of runtime, memory usage, unpromising pattern filtering, and scalability.
\end{enumerate}

The rest of this paper is organized as follows. Related work is briefly reviewed in Section \ref{sec:relatedwork}. Preliminaries and the problem statement of high-utility sequential pattern mining are presented in Section \ref{sec:problem}. The proposed HUSP-ULL algorithm with the UL-list and two pruning strategies are presented in Section \ref{sec:algorithm}. An experimental evaluation of the designed algorithms is provided in Section \ref{sec:experiment}.  Finally, a conclusion is presented and  some opportunities for future work  are described in Section \ref{sec:conclusion}.

\section{Related Work}
\label{sec:relatedwork}

We structure the related work around the two main elements that this paper addresses: high-utility itemset mining and high-utility sequential pattern mining. 

\subsection{High-Utility Itemset Mining}

The problem of high-utility itemset mining (HUIM) \cite{chan2003mining,yao2004foundational} was designed to find the set of high-utility itemsets (HUIs),  i.e. the itemsets that have utility values that are greater than or equal to a minimum utility threshold.  Since HUIM does not provide a downward closure property to reduce the search space, unlike association rule mining (ARM) \cite{agrawal1994fast}, it is necessary to find other strategies for reducing the search space. To obtain a downward closure property that can be used in HUIM, Liu et al. \cite{liu2005two} designed a property called the transaction-weighted downward closure (TWDC)  and defined a set of candidate itemsets called the high transaction-weighted utilization itemsets (HTWUIs).  Based on the TWU and HTWUIs, Liu et al. designed an algorithm named Two-Phase to mine HUIs. This latter first discovers the set of HTWUIs using a breadth-first search and then select HUIs in the discovered HTWUIs. To obtain better performance for mining HUIs, tree-based HUIM algorithms were introduced such as IHUP \cite{ahmed2009efficient}, UP-Growth \cite{tseng2010up} and UP-Growth$ ^{+} $ \cite{tseng2013efficient}.  To reduce the number of candidates, Liu et al. then~\cite{liu2012mining} proposed the HUI-Miner algorithm, which efficiently discovers  HUIs using a vertical structure called utility-list. This procedure identifies HUIs without generating candidates and without performing multiple database scans.

Up to now, the development of HUIM algorithms has been extensively studied, and many algorithms have been published to mine different kinds of HUIs in many real-life applications. On the one hand, some utility mining algorithms focus on the mining efficiency, such as FHM \cite{fournier2014fhm}, EFIM \cite{zida2015efim} and d$ ^{2} $HUP \cite{liu2012direct}. On the other hand, many models and algorithms mainly aim at the study of utility-oriented mining effectiveness. For example, discovering various kinds of HUIs such as mining HUIs in uncertain databases \cite{2lin2016efficient}, mining the top-\textit{k} HUIs without setting the minimum utility threshold \cite{tseng2016efficient}, exploiting non-redundant correlated utility patterns \cite{lin2017fdhup,gan2017extracting}, extracting the up-to-date HUIs which can show the trends \cite{lin2015efficient}, and mining on-shelf HUIs from the temporal databases \cite{lan2011discovery}. Yun et al. \cite{yun2018damped} proposed a damped window to extract high average utility patterns over data streams. In contrast to static data, the dynamic data is more complex and desirable in many real-life applications. Several dynamic utility mining models \cite{yun2017efficient,gan2018survey1} have been proposed to deal with dynamic data.

\subsection{High-Utility Sequential Pattern Mining}

Sequential pattern mining (SPM) \cite{agrawal1995mining,han2000freespan,pei2001prefixspan,srikant1996mining} is important as it considers the sequential ordering of itemsets, which is important for applications such as behavior analysis, DNA sequence analysis, and weblog mining.  SPM was proposed by Agrawal \cite{agrawal1995mining} and has been extensively studied. Several efficient algorithms have been developed  such as GSP \cite{srikant1996mining}, FreeSpan \cite{han2000freespan}, PrefixSpan \cite{pei2001prefixspan}, SPADE \cite{zaki2001spade} and SPAM \cite{ayres2002sequential}. Several recent literature surveys of the development of SPM can be further referred to \cite{fournier2017survey,3gan2018survey}. Traditional SPM algorithms rely on the frequency/support framework for discovering frequent sequences, which does not take business interests into account. 
High-utility sequential pattern mining (HUSPM) \cite{alkan2015crom,wang2016efficiently,yin2012uspan} was developed by combining HUIM and SPM to mine high-utility sequential patterns (HUSPs). It was first used for mining high-utility path traversal patterns of web pages \cite{zhou2007utility}, high-utility web access sequences \cite{ahmed2010mining}, and high-utility mobile sequential patterns \cite{shie2013efficient,shie2011mining}. However, the above algorithms can only handle simple sequences. Ahmed et al. \cite{ahmed2010novel} designed a level-wise approach called UL and a pattern-growth approach named US for HUSPM. HUSPM takes ordered sequences as input and reveals sequential patterns having high utilities, which has been a challenging and important issue in recent decades.
Hence, Yin et al. \cite{yin2012uspan} proposed a formal framework for HUSPM and introduced an efficient USpan algorithm to discover HUSPs. Information about the utility of each node in the tree is stored in a developed matrix for mining HUSPs without performing multiple database scans. Two pruning strategies based on the sequential-weighted downward closure property and on the remaining utility model were designed to reduce the search space for mining HUSPs. However, USpan may fail to discover the complete HUSPs due to the used upper bound \cite{gan2019proum}.

To facilitate parameter setting for HUSP mining, Yin et al. proposed the TUS algorithm, which discovers the top-\textit{k} HUSPs. Lan et al. \cite{lan2014applying} then proposed a projection-based HUSP approach with uses a sequence-utility upper-bound (\textit{SUUB}) to mine high-utility sequential patterns using the maximum utility measure. The maximum utility measure was designed to find a small set of patterns, which provides rich information about HUSPs. A novel indexing strategy and a sequence utility table containing actual utilities and upper-bound on utilities of candidates was used to improve the mining performance. Then, to further improve the mining performance, Alkan et al. \cite{alkan2015crom} proposed the high-utility sequential pattern extraction (HuspExt) algorithm. It calculates a Cumulated Rest of Match (\textit{CRoM}) to obtain  an upper-bound on utility values and prune unpromising candidates early. Recently, Wang et al. \cite{wang2016efficiently} developed two tight utility upper-bounds, named prefix extension utility (\textit{PEU}), and reduced sequence utility (\textit{RSU}) to speed up the discovery of HUSPs. An efficient HUS-Span algorithm was further proposed to mine HUSPs and a TKHUS-Span algorithm was also developed to identify the top-\textit{k} HUSPs \cite{wang2016efficiently}. Gan et al. \cite{gan2019proum} proposed an efficient projection-based utility mining approach named ProUM to discover high-utility sequences by utilizing the upper bound of sequence extension utility (\textit{SEU}) and the utility-array structure. 

Wu et al. \cite{wu2013mining} proposed a model to extract high-utility episodes in complex event sequences. In order to get insights that facilitate decision making for expert and intelligent systems, Lin et al. introduced the utility-based episode rules for investment \cite{lin2015discovering}. Recently, an incremental model for HUSP mining is introduced in \cite{wang2018incremental}. The comprehensive review of utility-oriented pattern mining can be referred to \cite{gan2018survey1,gan2018survey,gan2018privacy}.

\section{Preliminaries and Problem Statement}
\label{sec:problem}
In this section, we introduce notations and concepts used in the paper. Then, we give formal problem definition.

\subsection{Notations and Concepts}

Let \textit{I} = \{$ i_{1} $, $ i_{2} $, $\dots$, $ i_{m} $\} be a finite set of distinct items (symbols). 
A quantitative itemset, denoted as \textit{v} = [($ i_{1}, q_{1} $) ($ i_{2}, q_{2} $) $\dots$ ($ i_{c}, q_{c} $)], is a subset of \textit{I} and each item in a quantitative itemset is associated with a quantity (internal utility). 
An itemset, denoted as $ w $ = [$ i_{1} $, $ i_{2} $, $\dots$, $ i_{c} $], is a subset of $ I $ without quantities. Without loss of generality, we assume that items in an itemset (quantitative itemset) are listed in \textit{alphabetical} order since items are unordered in an itemset (quantitative itemset). 
A quantitative sequence is an ordered list of one or more quantitative itemsets, which is denoted as $ s $ = $ < $$ v_{1} $, $ v_{2} $, $\dots$, $ v_{d} $$ > $. 
A sequence is an ordered list of one or more itemsets without quantities, which is denoted as $ t $ = $ < $$ w_{1} $, $ w_{2} $, $\dots$, $ w_{d} $$ > $.

For convenience, in the following "quantitative" will be abbreviated as  "\textit{q}-". 
Thus, the term "\textit{q}-sequence" will be used to refer to a sequence with quantities, and "sequence" to refer to sequences without quantities. 
Similarly, a "\textit{q}-itemset" is an itemset having quantities, while "itemset" refers to an itemset that does not have quantities. 
For example, $ < $[(\textit{a}, 2) (\textit{b}, 1)], [(\textit{c}, 3)]$ > $ is a \textit{q}-sequence while $ < $[\textit{ab}], [\textit{c}]$ > $ is a sequence. [(\textit{a}, 2) (\textit{b}, 1)] is a \textit{q}-itemset and [\textit{ab}] is an itemset. A quantitative sequential database is a set of transactions \textit{D} = \{$ S_{1} $, $ S_{2} $, $\dots$, $ S_{n} $\}, where each transaction $ S_{q}\in D $ is a \textit{q}-sequence, and has a unique identifier \textit{q} called its \textit{SID}. In addition, each item in \textit{D} is associated with a profit (external utility), which is denoted as $ pr(i_{j}) $.

Consider the following running example. A quantitative sequential database is shown in Table \ref{table:db}. This database has 6 transactions and 6 items. Table \ref{table:profit} is a utility table that provides a unit profit for each item of Table \ref{table:db}. 
In the running example, [(\textit{a}:2)(\textit{c}:3)] is the first \textit{q}-itemset of transaction $ S_{1} $. The quantity of an item (\textit{a}) in this \textit{q}-itemset is 2, and its utility is calculated as $ 2 \times 5 = 10 $. 

\begin{table}[!htbp]
	\setlength{\abovecaptionskip}{0pt}
	\setlength{\belowcaptionskip}{0pt} 
	\caption{A Quantitative Sequential Database.}
	\centering
	\begin{tabular}{|c|c|}
		\hline
		\textbf{SID}	 & \textbf{Q-sequence} \\ \hline
		$ S_{1} $ & $ < $[(\textit{a}:2)(\textit{c}:3)], [(\textit{a}:3)(\textit{b}:1)(\textit{c}:2)], [(\textit{a}:4)(\textit{b}:5)(\textit{d}:4)], [(\textit{e}:3)]$ > $  \\ \hline
		$ S_{2} $ & $ < $[(\textit{a}:1)(\textit{e}:3)], [(\textit{a}:5)(\textit{b}:3)(\textit{d}:2)], [(\textit{b}:2)(\textit{c}:1)(\textit{d}:4)(\textit{e}:3)]$ > $  \\ \hline
		$ S_{3} $ & $ < $[(\textit{e}:2)], [(\textit{c}:2)(\textit{d}:3)], [(\textit{a}:3)(\textit{e}:3)], [(\textit{b}:4)(\textit{d}:5)]$ > $  \\ \hline
		$ S_{4} $ & $ < $[(\textit{b}:2)(\textit{c}:3)], [(\textit{a}:5)(\textit{e}:1)], [(\textit{b}:4)(\textit{d}:3)(\textit{e}:5)]$ > $  \\ \hline
		$ S_{5} $ & $ < $[(\textit{a}:4)(\textit{c}:3)], [(\textit{a}:2)(\textit{b}:5)(\textit{c}:2)(\textit{d}:4)(\textit{e}:3)]$ > $  \\ \hline
		$ S_{6} $ & $ < $[(\textit{f}:4)], [(\textit{a}:5)(\textit{b}:3)], [(\textit{a}:3)(\textit{d}:4)]$ > $  \\ \hline
	\end{tabular}
	\label{table:db}
\end{table}

\begin{table}[!htbp] 
	\setlength{\abovecaptionskip}{0pt}
	\setlength{\belowcaptionskip}{0pt} 		
	\caption{An Utility Table.}
	\centering
	\begin{tabular}{|c|c|c|c|c|c|c|}
		\hline
		\textbf{Item} &	\textit{a} & \textit{b} & \textit{c} &	\textit{d} & \textit{e} & \textit{f} \\ \hline
		\textbf{Profit (\$)} & 5 & 3 & 4 & 2 & 1 & 6 \\ \hline
	\end{tabular}
	\label{table:profit}	
\end{table}

\begin{definition}
	The utility of an item ($ i_{j} $) in a \textit{q}-itemset \textit{v} is denoted as $ u(i_{j}, v) $, and defined as:
	\begin{equation}
	u(i_{j}, v) = q(i_{j}, v)\times pr(i_{j}),
	\end{equation}
	where $ q(i_{j}, v) $ is the quantity of ($ i_{j} $) in $ v $, and $ pr(i_{j}) $ is the profit of ($ i_{j} $).
\end{definition}

For instance, the utility of item (\textit{c}) in the first \textit{q}-itemset of $ S_{1} $ in Table \ref{table:db} is calculated as: $ u(c, [(a:2) (c:3)]) $ = $ q(c, [(a:2) (c:3)])\times pr(c)$ = 3 $\times$ \$4 = \$12.

\begin{definition}
	The utility of a \textit{q}-itemset $ v $ is denoted as $ u(v) $ and defined as:
	\begin{equation}
	u(v) = \sum_{i_{j}\in v}u(i_{j}, v).
	\end{equation}
\end{definition}

For example in Table \ref{table:db}, \textit{u}([(\textit{a}:2)(\textit{c}:3)]) = \textit{u}(\textit{a}, [(\textit{a}:2)(\textit{c}:3)]) + \textit{u}(\textit{c}, [(\textit{a}:2)(\textit{c}:3)]) = 2 $\times$ \$5 + 3 $\times$ \$4 = \$22.

\begin{definition}	
	The utility of a \textit{q}-sequence $ s $ = $<$$v_{1}, v_{2}, \dots, v_{d}$$>$ is defined as:
	\begin{equation}
	u(s) = \sum_{v\in s}u(v).
	\end{equation}
\end{definition}

For instance, consider  Table \ref{table:db}. We have that $ u(S_{1}) $ = \textit{u}([(\textit{a}:2)(\textit{c}:3)]) + \textit{u}([(\textit{a}:3)(\textit{b}:1)(\textit{c}:2)]) + \textit{u}([(\textit{a}:4)(\textit{b}:5)(\textit{d}:4)]) + \textit{u}([(\textit{e}:3)]) = \$22 + \$26 + \$43 + \$3 = \$94.

\begin{definition}
	The utility of a quantitative sequential database \textit{D} is the sum of the utility of each of its q-sequences:
	\begin{equation}
	u(D) = \sum_{s\in D}u(s).
	\end{equation}			
\end{definition}

For example,  $ u(D) $ = $ u(S_{1}) $ + $ u(S_{2}) $ + $ u(S_{3}) $ + $ u(S_{4}) $ + $ u(S_{5}) $ + $ u(S_{6}) $ = \$94 + \$67 + \$56 + \$67 + \$76 + \$81 = \$441, as shown in Table \ref{table:db}.

\begin{definition}
	Given a \textit{q}-sequence \textit{s} = $<$$v_{1}, v_{2}, \dots, v_{d}$$>$ and a sequence $ t $ = $<$$w_{1}, w_{2}, \dots, w_{d'}$$>$, if $ d $ = $ d' $ and the items in $ v_{k} $ are the same as the items in $ w_{k} $ for $ 1\leq k\leq d $, $ t $ matches $ s $, which is denoted as $ t\sim s $.	
\end{definition}

For instance, in Table \ref{table:db}, $<$[\textit{ac}], [\textit{abc}], [\textit{abd}], [\textit{e}]$>$ matches $ S_{1} $. Note that it is possible that a sequence has more than one match in a $q$-sequence. For instance, $<$[\textit{a}],[\textit{b}]$>$ has three matches as $<$[\textit{a}:2],[\textit{b}:1]$>$, $<$[\textit{a}:2],[\textit{b}:5]$>$ and $<$[\textit{a}:3],[\textit{b}:5]$>$ in $ S_{1} $. Because of this, HUSP is generally considered  as more challenging than SPM and HUIM.

\begin{definition}
	Let there be some itemsets $ w $ and $ w' $. The itemset $ w $ is contained in $ w' $ (denoted as $ w\subseteq w' $) if $ w $ is a subset of $ w' $ or $ w $ is the same as $ w' $. 
	Given two \textit{q}-itemsets $ v $ and $ v' $, $ v $ is said to be contained in $ v' $ if for any item in $ v $, there exists the same item having the same quantity in $ v' $. This is denoted as $ v\subseteq v' $. 
\end{definition}

For example,  the itemset [\textit{ac}] is contained in the itemset [\textit{abc}] in Table \ref{table:db}. 
The \textit{q}-itemset [(\textit{a}:2)(\textit{c}:3)] is contained in [(\textit{a}:2)(\textit{b}:1)(\textit{c}:3)] and [(\textit{a}:2)(\textit{c}:3)(\textit{e}:2)], but [(\textit{a}:2)(\textit{c}:3)] is not contained in [(\textit{a}:2)(\textit{b}:3)(\textit{c}:1)] and [(\textit{a}:4)(\textit{c}:3)(\textit{d}:4)].

\begin{definition}
	Let there be some sequences \textit{t} = $<$$w_{1}, w_{2}, \dots, w_{d}$$>$ and $ t' $ = $<$$w'_{1}, w'_{2}, \dots, w'_{d'}$$>$. The sequence $ t $ is contained in $ t' $ (denoted as $ t\subseteq t $') if there exists an integer sequence $ 1\leq k_{1}\leq k_{2}\leq\dots\leq d' $ such that $ w_{j}\subseteq w'_{k_{j}} $ for $ 1\leq j\leq d $. 
	Let there be two \textit{q}-sequences $ s $ = $<$$v_{1}, v_{2}, \dots, v_{d}$$>$ and $ s' $ = $<$$v'_{1}, v'_{2}, \dots, v'_{d'}$$>$.  $ s $ is said to be contained in $ s' $ (denoted  as $ s\subseteq s' $) if there exists an integer sequence $ 1\leq k_{1}\leq k_{2} \leq\dots\leq d' $ such that $ v_{j}\subseteq v'_{k_{j}} $ for $ 1\leq j \leq d $.
	In the rest of this paper, $ t \subseteq s $ will be used to indicate that $ t \sim s_{k} \wedge s_{k} \subseteq s $ for convenience.
\end{definition}	

For example in Table \ref{table:db}, $<$[(\textit{a}:2)],[(\textit{e}:3)]$>$ and $<$[(\textit{a}:4)],[(\textit{e}:3)]$>$ are contained in $ S_{1} $, but $<$[(\textit{a}:1)],[\textit{e}:3]$>$ and $<$[(\textit{a}:4)],[(\textit{e}:4)]$>$ are not contained in $ S_{1} $. 

\begin{definition}
	A \textit{k}-itemset, also called  \textit{k}-\textit{q}-itemset
	is an itemset that contains exactly \textit{k} items. 
	A \textit{k}-sequence (\textit{k}-\textit{q}-sequence) is a sequence having \textit{k} items. 
\end{definition}	

Consider the database of Table \ref{table:db}. The $q$-sequence $ S_{1} $ is a 9-\textit{q}-sequence. Its first \textit{q}-itemset is a 2-\textit{q}-itemset.

\begin{definition}
	Let there be a sequence $t$ and a $q$-sequence $s$. 	The utility of  $ t $ in \textit{s} is defined as:
	\begin{equation}
	u(t, s) = max\{u(s_{k})|t \sim s_{k} \wedge s_{k} \subseteq s\}.
	\end{equation}
\end{definition}	

For instance, for the sequential database of Table \ref{table:db}, $ u(<[a],[b]>, S_{1}) $ = $ max$\{\textit{u}($<$[\textit{a}:2],[\textit{b}:1]$>$), \textit{u}($<$[\textit{a}:2],[\textit{b}:5]$>$), \textit{u}($<$[\textit{a}:3],[\textit{b}:5]$>$)\} = $ max $\{\$13, \$25, \$30\} = \$30. 
In this example, it can be seen that several utility values can be associated to a pattern in a same \textit{q}-sequence. This is different from traditional SPM and HUIM.

\begin{definition}
	The utility of a sequence $ t $ in a quantitative sequential database $ D $ is denoted as $ u(t) $ and defined as:
	\begin{equation}
	u(t) = \sum_{s\in D} \{u(t,s) | t \subseteq s\}.
	\end{equation}
\end{definition}

For example in Table \ref{table:db}, \textit{u}($<$[\textit{a}],[\textit{b}]$>$) = \textit{u}($<$[\textit{a}],[\textit{b}]$>$, $S_{1}$) + \textit{u}($<$[\textit{a}],[\textit{b}]$>$, $S_{2}$) + \textit{u}($<$[\textit{a}],[\textit{b}]$>$, $S_{3}$) + \textit{u}($<$[\textit{a}],[\textit{b}]$>$, $S_{4}$) + \textit{u}($<$[\textit{a}],[\textit{b}]$>$, $S_{5}$) = \$30 + \$31 + \$27 + \$37 + \$35 = \$160.

\subsection{Problem Definition}

\begin{definition}[High-Utility Sequential Pattern, HUSP]
	A sequence $ t $ in a quantitative sequential database \textit{D} is defined as a high-utility sequential pattern (denoted as \textit{HUSP}) if its total utility is no less than the minimum utility threshold $\delta$:
	\begin{equation}
	HUSP\gets\{t|u(t)\geq \delta \times u(D)\}.
	\end{equation}
\end{definition}

For example in Table \ref{table:db}, \textit{u}($ < $[\textit{a}],[\textit{b}]$ > $)(= \$160). If $ \delta = 0.1$, then $ < $[\textit{a}],[\textit{b}]$ > $ is a HUSP since \textit{u}($ < $[\textit{a}],[\textit{b}]$ > $) (= \$160) $ > $ $ \delta \times u(D) $ (= \$44.1).

\textbf{Problem Statement:} Based on the above concepts, the formal definition of the problem studied in this work  is defined below. Let there be a quantitative sequential database and a user-defined minimum utility threshold. High-utility sequential pattern mining (HUSPM) consists of enumerating all HUSPs whose total utility values in this database are no less than the minimum utility threshold. 

Therefore, the objective of high-utility sequential pattern mining is to identify sequential patterns that contain total utility in a sequence database that meets or exceeds a prespecified minimum utility threshold. These insightful profitable sequential patterns can be used in some specific applications, such as market basket analysis, E-commerce recommendation with personalized promotion, click-stream analysis, improve service quality using the searching/browsing/buying behavior, scenic route recommendation by optimizing user-specified multi-preferences (e.g., utility, safety, cost, travel distance or time).

\section{The Proposed HUSP-ULL Algorithm}
\label{sec:algorithm}

This section presents a novel algorithm named HUSP-ULL for the problem of high-utility sequential pattern mining (HUSPM). The HUSP-ULL algorithm first scans the database to find  1-sequences and build a lexicographic sequence (LS)-tree. The LS-tree is a representation of the search space used for mining HUSPs. Details of the LS-tree, utility-linked (UL)-list, pruning strategies, and the main procedure of the HUSP-ULL algorithm are respectively explained in this section.

\subsection{Concatenations and Lexicographic Sequence Tree}

Each node in a lexicographic sequence (LS)-tree \cite{ayres2002sequential} represents a candidate HUSP, whose utility can be compared with the minimum utility threshold to determine if the candidate is a HUSP. For each node that the algorithm visits in the LS-tree, a  projected database is built, which consists of utility-linked (UL)-lists obtained by transforming   transactions  ($q$-sequences) of the original database. The algorithm utilizes the UL-lists of each node (candidate HUSP) to calculate  its utility and upper-bounds. Each UL-list  represents a transaction (\textit{q}-sequence). To generate new sequences  (child nodes) of a node in the LS-tree, the designed algorithm performs two common operations \cite{pei2001prefixspan,yin2012uspan,wang2016efficiently} of sequence mining, called \textit{I}-\textit{Concatenation} and \textit{S}-\textit{Concatenation}, respectively. 

\begin{definition}[\textbf{\textit{I}-\textit{Concatenation}} and \textbf{\textit{S}-\textit{Concatenation}} \cite{pei2001prefixspan,yin2012uspan,wang2016efficiently}]
	Given a sequence $ t $ and an item $ i_{j} $, the \textit{I-Concatenation} of $ t $ with $ i_{j} $ consists of appending $ i_{j} $ to the last itemset of $ t $, denoted as $<$$t \oplus i_{j}$$>$$_{I-Concatenation}$. 
	The \textit{S-Concatenation} of $ t $ with an item $ i_{j} $ consists of adding $ i_{j} $ to a new itemset appended after the last itemset of $ t $, denoted as $<$$t \oplus i_{j}$$>$$_{S-Concatenation}$.
\end{definition}

For example, given a sequence $ t $ = $<$[$a$], [$b$]$>$ and a new item $(c)$, $<$$t \oplus c$$>$$_{I-Concatenation}$ = $<$[\textit{a}],[\textit{bc}]$>$ and $<$$t \oplus c$$>$$_{S-Concatenation}$ = $<$[\textit{a}],[\textit{b}],[\textit{c}]$>$. 
Based on the previous definitions, it follows that the number of itemsets in $ t $ does not change after performing an \textit{I}-\textit{Concatenation}, while performing an \textit{S-Concatenation} increases the number of itemsets in $ t $  by one. Based on the two operations, all candidates of the search space can be generated for the purpose of mining HUSPs.

The search process of the proposed algorithm can be viewed as the process of building a LS-tree step-by-step. The proposed algorithm initially scans the database to identify the set of 1-sequences that satisfy the minimum utility threshold. The LS-tree is then explored starting from 1-sequences by performing a depth-first search. The child nodes of a given node are obtained by performing the \textit{I}-\textit{Concatenation} or \textit{S}-\textit{Concatenation} operations on that node. The LS-tree is essentially an enumeration tree that is used to list the complete set of sequences. Each node in the tree represents a sequence. Fig. \ref{fig:tree} shows a partial LS-tree built based on the database of Table \ref{table:db}.

\begin{figure}[!htbp]
	\setlength{\abovecaptionskip}{0pt}
	\setlength{\belowcaptionskip}{0pt} 		
	\centering
	\includegraphics[width=3.6in]{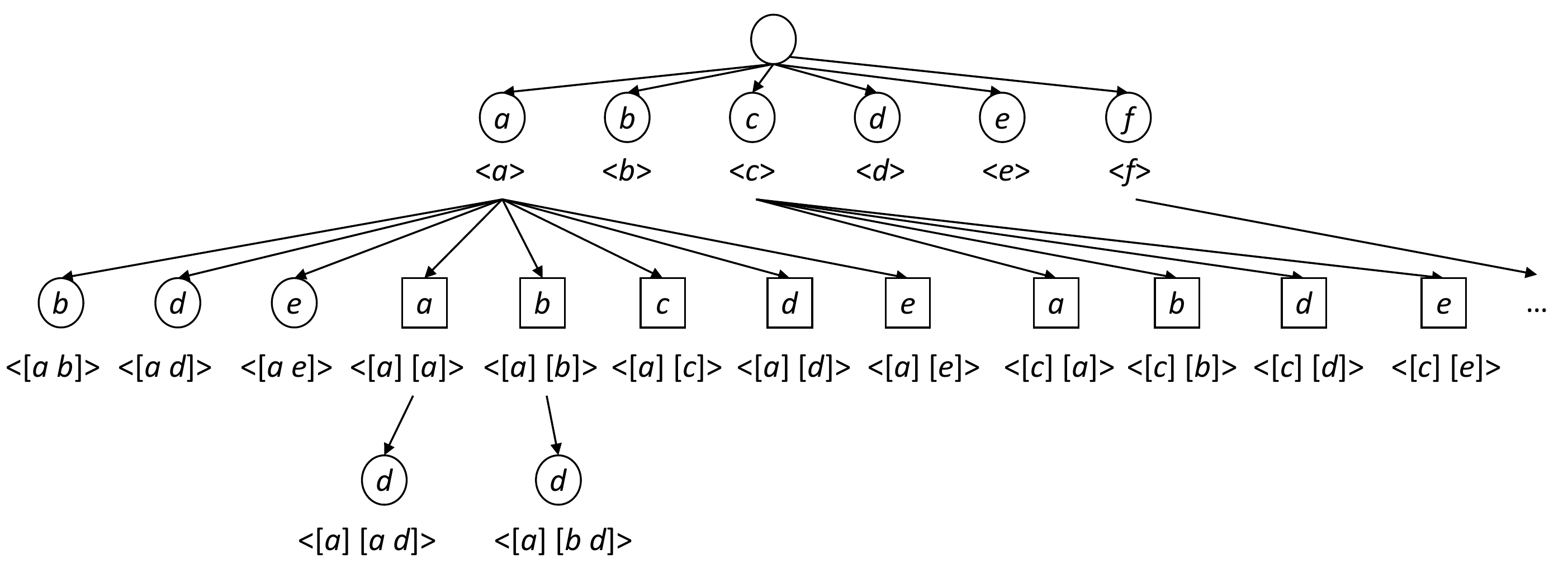}
	\caption{A lexicographic sequence (LS)-tree.}
	\label{fig:tree}
\end{figure}

In Fig. \ref{fig:tree}, 1-sequences such as $ < $\textit{a}$ > $, $ < $\textit{b}$ > $, and $ < $\textit{c}$ > $, are children of the root. Circles are used to denote patterns obtained by performing an  \textit{I}-\textit{Concatenation}, while squares denotes sequences obtained by performing an \textit{S}-\textit{Concatenation}. Notice that each LS-tree node represents a candidate of the search space of HUSPs.

To ensure the completeness and correctness for mining HUSPs, an order is defined for processing sequences. Let there be two sequences $ t_{a} $ and $ t_{b} $. It is said that $ t_{a} \prec t_{b} $ 
if 1) The length of $ t_{a} $ is less than that of $ t_{b} $; 
2) $ t_{a} $ is obtained by an \textit{I}-\textit{Concatenation} on a sequence $ t $ while $ t_{b} $ is  obtained by an \textit{S}-\textit{Concatenation} on a sequence $ t $; 
3) $ t_{a} $ and $ t_{b} $ are both obtained by respectively performing an \textit{I}-\textit{Concatenation} or \textit{S}-\textit{Concatenation} on a sequence \textit{t},  
and the  item added to $ t_{a} $ is lexicographically smaller than the one added to $ t_{b} $. 
This order on sequences is also applied to $ q $-sequences. For example, $<$[\textit{a}]$>$ $\prec$ $<$[\textit{ab}]$>$ $\prec$ $<$[\textit{a}],[\textit{a}]$>$ $\prec$ $<$[\textit{a}],[\textit{c}]$>$. To discover the complete set of HUSPs, the designed algorithm enumerates all candidates by performing the two concatenation operations following that processing order.

\subsection{The Utility-Linked List Structure}

To calculate the utility and upper-bound values of candidates, the designed algorithm could scan transactions from the original database. However, this process would result in long execution times because there are often multiple matches in a sequence.
To handle this situation, the compact utility-linked (UL)-list structure is introduced to store  information about the utility of each transaction. 
This structure is used to efficiently generate the utility of sequences obtained by \textit{I}-\textit{Concatenations} and \textit{S}-\textit{Concatenations} to continue the search for patterns. To illustrate the concept of UL-list structure, an example is provided in Table \ref{table:ul-list} using the utility-table and database given before. In Table \ref{table:ul-list}, it is the UL-list structure of the sequence $ S_{1} $ in Table \ref{table:db}.

\begin{table}[!htbp] 
	\setlength{\abovecaptionskip}{0pt}
	\setlength{\belowcaptionskip}{0pt}
	\caption{The Utility-Linked (UL)-List Structure of $ S_{1} $.}
	\label{table:ul-list}
	\centering
	\begin{tabular}{|c|c|}
		\hline
		\textbf{UP Information} & $ < $[(\textit{a}, \$10, \$84, 3) (\textit{c}, \$12, \$72, 5)], \\&[(\textit{a}, \$15, \$57, 6) (\textit{b}, \$3, \$54, 7) (\textit{c}, \$8, \$46, -)], \\&[(\textit{a}, \$20, \$26, -) (\textit{b}, \$15, \$11, -) (\textit{d}, \$8, \$3, -)], \\&[\textit{e}, \$3, \$0, -]$ > $ \\ 
		\hline
		\textbf{Header Table} &	(\textit{a}, 1) (\textit{b}, 4) (\textit{c}, 2) (\textit{d}, 8) (\textit{e}, 9)\\ \hline 
	\end{tabular}
\end{table}

The UL-list structure contains two arrays, \textbf{Header Table} and \textbf{UP (utility and position) Information}. Details are described below. 

\textbf{1) Header Table}. The Header Table in the UL-list structure stores  a set of distinct items with their first occurrence positions in the transformed transaction. For example in Table \ref{table:ul-list}, the distinct items of $ S_{1} $ are (\textit{a}), (\textit{b}), (\textit{c}), (\textit{d}), and (\textit{e}) and their first occurrence positions in $ S_{1} $ are respectively 1, 4, 2, 8 and 9. 

\textbf{2) UP Information}. In terms of information about UP (utility and position) of each sequence, each element respectively stores the \textbf{\underline{item name}}, the \textbf{\underline{utility}} \textbf{\underline{of the item}}, the \textbf{\underline{remaining utility of the item}}, and the \textbf{\underline{next position of the item}}. For example in Table \ref{table:ul-list}, the utility of the item (\textit{a}) in the first element is calculated as 10 in $ S_{1} $; the total utility excluding the item (\textit{a}) in $ S_{1} $ (called \textit{remaining utility}) is calculated to be 84, and the next position of the item (\textit{a}) in $ S_{1} $ is found to be 3. For each node in the LS-tree, transactions containing this node (sequence) are transformed into a utility-linked (UL)-list and attached to the projected database of this node. The utilities and upper-bound values of the candidates can be easily calculated from the projected database using the UL-list structure.

As mentioned, the UL-list structure can be used to calculate the utilities and the upper-bound values of candidates for deriving all HUSPs. However, a sequence may have multiple matches in a \textit{q}-sequence, and hence a sequence may have multiple utilities in a \textit{q}-sequence. Thus, it is necessary to find the positions of the matches to calculate the utilities and the upper-bound values of the processed node (sequence).

For convenience, the position of the last item within each match is defined as the \textbf{\underline{concatenation point}}, and the first concatenation point is called the \textbf{\underline{start point}}. 
For example, consider the database of Table \ref{table:db}. The sequence $ t $ = $<$[\textit{a}],[\textit{b}]$>$ has three matches in $ S_{1} $, that is $<$[\textit{a}:2],[\textit{b}:1]$>$, $<$[\textit{a}:2],[\textit{b}:5]$>$ and $<$[\textit{a}:3],[\textit{b}:5]$>$. 
The concatenation points of $ t $ in $ S_{1} $  are 4, 7 and 7, respectively, and the start point is 4. 
By definition,  an \textit{I-Concatenation} appends an item to the last itemset of a sequence. 
Thus, the candidate items for \textit{I-Concatenation} are the items appearing in the  itemsets containing concatenation points. 
In the above example, the candidate items for \textit{I-Concatenation} are \{(\textit{c}:2),(\textit{d}:4)\}. 
By definition, an \textit{S-Concatenation} adds an item to a new itemset, appended at the end of a sequence. 
Thus, in each transaction, the items in the itemsets appearing after the start point are  candidate items for \textit{S}-\textit{Concatenation}. 
In the above example, the start point (= 4) is in the second itemset. Hence, the items appearing after the second itemset are candidate items for \textit{S-Concatenation}, that is \{(\textit{a}:4),(\textit{b}:5),(\textit{d}:4),(\textit{e}:3)\}.
Since there can be multiple matches of the sequence \textit{t} in a \textit{q}-sequence, the utility of $ t $ in that \textit{q}-sequence is defined as the largest utility value of $t$ in that sequence. 

From the above example, it can be seen that the UL-list structure can speed up the process of finding candidate items and calculating the utilities and upper-bound values of sequences. The designed algorithm stores only one copy of the original database as UL-lists.
Then, for each considered sequence, a projected database is created, which records only the position of the sequence in the original database. Thus, the designed algorithm does not consume a large amount of memory.

\subsection{The Downward Closure Property of Upper Bound}

Based on the LS-tree and UL-lists, the proposed HUSP-ULL algorithm can successfully identify the complete set of HUSPs using a depth-first search that applies the two concatenations operations.
However, this process can lead to exploring a very large number of candidates in the LS-tree, since there is a combinational explosion of the number of candidates in the mining process of HUSPs. 
Since the downward closure property, also called Apriori property~\cite{agrawal1994fast}, does not hold in high-utility sequential pattern mining, a new downward closure property must be introduced to be able to reduce the search space and efficiently find all HUSPs. 
To speed up the mining process and maintain the downward closure property, a sequence-weighted utilization (SWU) \cite{yin2012uspan} upper-bound was proposed to obtain a sequence-weighted downward closure (SWDC) property for mining HUSPs. This property can be used to greatly reduce the search space and  eliminate  unpromising candidates early. 

\begin{definition}
	\label{def:swu}
	The sequence-weighted utilization (\textit{SWU}) \cite{yin2012uspan} of a sequence $ t $ in a quantitative sequential database $ D $ is denoted as $ SWU(t) $ and defined as:
	\begin{equation}
	SWU(t) = \sum_{s\in D} \{u(s) | t \subseteq s\}.
	\end{equation}
\end{definition}

For example in Table \ref{table:db}, $ SWU $($<$\textit{a}$>$) = $ u(S_{1}) $ + $ u(S_{2}) $ + $ u(S_{3}) $ + $ u(S_{4}) $ + $ u(S_{5}) $ + $ u(S_{6}) $= \$94 + \$67 + \$56 + \$67 + \$76 + \$81 (= \$441) and $ SWU $($<$\textit{f}$>$) = $ u(S_{6})$ (= \$81).

\begin{theorem}[sequence-weighted downward closure property, SWDC property \cite{yin2012uspan}] 
	\label{theorem-swu}
	Given a quantitative sequential database $ D $ and two sequences $ t $ and $ t' $. If $ t\subseteq t' $, then:
	\begin{equation}
	SWU(t')\leq SWU(t).
	\end{equation}
\end{theorem} 

\begin{proof}
	Since $ t\subseteq t' $, $ SWU(t') = \displaystyle\sum_{s\in D} \{u(s) | t' \subseteq s\} \leq \displaystyle\sum_{s\in D} \{u(s) | t \subseteq s\} = SWU(t) $.
\end{proof}

\begin{theorem}
	\label{theorem-swu-upper-bound}
	Given a quantitative sequential database $ D $ and a sequence $ t $, it can be obtained  that:
	\begin{equation}
	u(t)\leq SWU(t).
	\end{equation}
\end{theorem}

\begin{proof}
	Since $u(t, s)\leq u(s) $, we can obtain that $ u(t) = \displaystyle\sum_{s\in D} \{u(t,s) | t \subseteq s\} \leq \displaystyle\sum_{s\in D} \{u(s) | t \subseteq s\} = SWU(t)$.
\end{proof}

The SWDC property and Theorem \ref{theorem-swu-upper-bound}  ensure that if the $ SWU $ of a sequence $ t $ is less than the minimum utility threshold,  the utility of $ t $ is also less than that threshold. 
Moreover, if that condition holds, the utilities of all the super-sequences of $ t $ are also less than that threshold. 
Thus, numerous unpromising candidates can be pruned using the \textit{SWU}. However, the $ SWU $ of a sequence \textit{t} is usually much larger than the actual utilities of \textit{t} and its super-sequences. 

To improve the performance of the designed algorithm, the remaining utility model \cite{yin2012uspan} is proposed in the USpan algorithm \cite{yin2012uspan}. However, it is not a real upper bound and cannot provide the complete mining results of utility mining, as reported in \cite{gan2019proum}. Thus, the concept of sequence extension utility (\textit{SEU}) \cite{gan2019proum} was proposed in the projection-based ProUM algorithm, and the details can be referred to \cite{gan2019proum}. To explain the concept of remaining utility and sequence extension utility, several concepts related to sequences and \textit{q}-sequences must be introduced.

\begin{definition}
	Given two $ q $-sequences $ s $ and $ s' $, if $ s\subseteq s' $, the extension of $ s $ in $ s' $ is said to be the rest of $ s' $ after $ s $, and is denoted as $<$$ s' $-$ s $$>_{rest}$. 
	Given a sequence $ t $ and a $ q $-sequence $ s $, if $ t \sim s_{k} \wedge s_{k} \subseteq s $ $ (t \subseteq s) $, the extension of $ t $ in $ s $ is the rest of $ s $ after $ s_{k} $, which is denoted as $<$$s$-$t$$>_{rest}$, where $ s_{k} $ is the first match of $ t $ in $ s $. 
\end{definition}

For example, given two \textit{q}-sequences $ s $ = $<$[\textit{a}:2],[\textit{b}:5]$>$ and $ S_{1} $ in Table \ref{table:db}, the extension of $ s $ in $ S_{1} $ is $<$$S_{1}$ - $s$$>$$_{rest}$ = $<$[(\textit{d}:4)],[(\textit{e}:3)]$>$. Consider a sequence $ t $ = $<$[\textit{a}],[\textit{b}]$>$. There exist three matches of $ t $ in $ S_{1} $. The first one is $<$[\textit{a}:2],[\textit{b}:1]$>$. Thus, $<$$S_{1}$ - $t$$>$$_{rest}$ = $<$[(\textit{c}:2)],[(\textit{a}:4)(\textit{b}:5)(\textit{d}:4)],[(\textit{e}:3)]$>$.

\begin{definition}
	The set of  extension items of a sequence $ t $ in a quantitative sequential database \textit{D} is denoted as $ I(t)_{rest} $ and defined as:
	\begin{equation}
	I(t)_{rest} = \{i_{j}|i_{j}\in <s - t>_{rest}\wedge t \subseteq s \wedge s\in D\}.
	\end{equation}
\end{definition}

In the above example, \textit{I}($<$[\textit{a}],[\textit{b}]$>$)$_{rest}$ = \{$ a, b, c, d, e $\}.

\begin{definition}
	The sequence extension utility (\textit{SEU}) \cite{gan2019proum} of a sequence $ t $ in a quantitative sequential database $ D $ is denoted as  \textit{SEU}$(t) $ and defined as:
	\begin{equation}
	SEU(t) = \sum_{s\in D} \{u(t,s) + u(<s - t>_{rest})|t \subseteq s\}.
	\end{equation}
\end{definition}

Notice that \textit{u}($<$\textit{s} - \textit{t}$>_{rest}$) is the remaining utility of $ t $ in $ s $, which is the fourth element in the designed UL-list. For example in Table \ref{table:db}, consider the sequence $ t $= $<$[\textit{a}],[\textit{b}]$>$. Then, \textit{SEU}$(t) $ = $ u(t, S_{1})$ + \textit{u}($<$$S_{1}$ - $t$$>$$_{rest}$) + $ u(t, S_{2})$ + \textit{u}($<$$S_{2}$ - $t$$>$$_{rest}$) + $ u(t, S_{3})$ + \textit{u}($<$$S_{3}$ - $t$$>$$_{rest}$) + $ u(t, S_{4})$ + \textit{u}($<$$S_{4}$ - $t$$>$$_{rest}$) + $ u(t, S_{5})$ + \textit{u}($<$$S_{5}$ - $t$$>$$_{rest}$) = \$30 + \$54 + \$31 + \$25 + \$27 + \$10 + \$37 + \$11 + \$35 + \$19 = \$279.

\begin{theorem}
	\label{theorem-pu}
	Given a quantitative sequential database $ D $ and two sequences $ t $ and $ t' $. If $ t\subseteq t' $, we can obtain that: 
	\begin{equation}
	SEU(t')\leq SEU(t).
	\end{equation}		
\end{theorem}

\begin{theorem}
	\label{theorem-pu-upper-bound}
	Given a quantitative sequential database $ D $ and a sequence $ t $, it follows that:
	\begin{equation}
	u(t)\leq SEU(t).
	\end{equation}	
\end{theorem}

Proof of Theorems \ref{theorem-pu} and \ref{theorem-pu-upper-bound} can be referred to \cite{gan2019proum}. They indicate that for a sequence $ t $, if \textit{SEU}$(t) $ is less than the minimum utility threshold and the utility of $ t $ is less than that threshold, the utilities of the super-sequences of $ t $ are less than that threshold.
If the \textit{SEU} or $ SWU $ of $ t $ is less than the threshold, the utility of $ t $ and the utilities of the super-sequences of $ t $ are less than the threshold, which indicates that $ t $ and the super-sequences of $ t $ are not HUSPs.
However,  the algorithm may still explore a large search space since the $ SWU $ and \textit{SEU}  upper-bounds are overestimations of utility values of patterns. To improve the mining performance and reduce the search space by pruning a large number of candidates, we introduce a tighter upper-bound for mining HUSPs, which is based on the PEU model~\cite{wang2016efficiently}. Details are given next.

\begin{definition}
	The prefix extension utility of a sequence $ t $ in a \textit{q}-sequence $ s $ is denoted as $ PEU(t,s) $ and defined as:
	\begin{equation}
	PEU(t, s) = max\{u(s_{k})+u(\textnormal{$ < $}s \textnormal{$ - $} s_{k}\textnormal{$ > $}_{rest})|t \sim s_{k} \wedge s_{k} \subseteq s\}.
	\end{equation}	
\end{definition}

For example, consider Table \ref{table:db} and a sequence \textit{t} = $<$[\textit{a}],[\textit{b}]$>$. This sequence has 3 matches in $ S_{2} $, which are $<$[\textit{a}:1],[\textit{b}:3]$>$, $<$[\textit{a}:1],[\textit{b}:2]$>$ and $<$[\textit{a}:5],[\textit{b}:2]$>$. Thus, we can obtain that $u$($<$$S_{2}$ - $<$[\textit{a}:1],[\textit{b}:3]$>>_{rest}$) = $u$($<$[(\textit{d}:2)],[(\textit{b}:2)(\textit{c}:1)(\textit{d}:4)(\textit{e}:3)]$>$) = \$25, $u$($<$$S_{2}$ - $<$[\textit{a}:1],[\textit{b}:2]$>>_{rest})$ = $u$($<$[(\textit{c}:1)(\textit{d}:4)(\textit{e}:3)]$>$) = \$15 and $u$($<$$S_{2}$ - $<$[\textit{a}:5],[\textit{b}:2]$>>_{rest}$) = $u$($<$[(\textit{c}:1)(\textit{d}:4)(\textit{e}:3)]$>$) = \$15. The utilities of the three matches are \$14, \$11 and \$31, respectively. Thus, $ PEU $($<$[\textit{a}],[\textit{b}]$>$,$S_{2}$) = \textit{max}\{\$14 + \$25, \$11 + \$15, \$31 + \$15\} = \$46.

\begin{definition}
	The prefix extension utility of a sequence $ t $ in $ D $ is denoted as $ PEU(t) $ and defined as:	
	\begin{equation}
	PEU(t) = \sum_{s\in D}\{PEU(t, s)|t\subseteq s\}.    
	\end{equation}
\end{definition}

For example, consider Table \ref{table:db} and the sequence \textit{t} = $<$[\textit{a}],[\textit{b}]$>$. The value $ PEU(t) $ is calculated as (\$67 + \$46 + \$37 + \$48 + \$54) = \$252, which is smaller than \textit{SEU}(\textit{t}) (= \$279).

\begin{theorem}
	\label{theorem:meu}
	Given a quantitative sequential database $ D $, and two sequences \textit{t} and \textit{t'}. If \textit{t}$ \subseteq $ \textit{t'}, we obtain that:
	\begin{equation}
	PEU(t') \leq PEU(t).	
	\end{equation}
\end{theorem}

\begin{proof}
	Suppose that $ s $ is a transaction in $ D $, which contains \textit{t} and \textit{t'}. 
	Let $ s_{q} $ be a \textit{q}-sequence satisfying \{$ u(s_{q} $) + \textit{u}($ < $\textit{s} - $ s_{q} $$ > $$ _{rest} $)\} = $ PEU(t, s) $, where $ t \sim s_{q} \wedge s_{q} \subseteq s $. Let $ s_{q'} $ be a \textit{q}-sequence satisfying \{$ u $($ s_{q'} $) + $ u $($ < $\textit{s} - $ s_{q'}>_{rest} $)\} = $ PEU(t', s) $ where $ t' \sim s_{q'} \wedge s_{q'} \subseteq s $. Since $ t\subseteq t' $,  we can divide $ t' $ into two parts as the prefix $ t $ and the extension $ e $ such that $ t $ + $ e $  = $ t' $. Similarly,  $ s_{q'} $ can  be divided into two parts as the prefix $ s_{q'_{t}} $ matching $ t $ and the extension $ s_{q'_{e}} $ matching $ e $ such that $ s_{q'_{t}} $ + $ s_{q'_{e}} $ = $ s_{q'} $. Thus, 
	\begin{align*}
	PEU(t', s) &= \{u(s_{q'}) + u(<s - s_{q'}>_{rest})\}\\
	&= \{u(s_{q'_{t}}) + u(s_{q'_{e}}) + u(<s - s_{q'}>_{rest})\}\\
	&\leq \{u(s_{q'_{t}}) + u(<s - s_{q'_{t}}>_{rest})\}\\
	&\leq \{u(s_{q}) + u(<s - s_{q}>_{rest})\} =  PEU(t, s).
	\end{align*}
	We can thus obtain that $\displaystyle PEU(t') = \sum_{s\in D}\{PEU(t', s)|t'\subseteq s\}\leq \sum_{s\in D}\{PEU(t, s)|t'\subseteq s\} \leq \sum_{s\in D}\{PEU(t, s)|t\subseteq s\} = PEU(t)$. 
\end{proof}

Theorem \ref{theorem:meu} indicates that if the $PEU$ value of a sequence $ t $ is less than the minimum utility threshold, the $PEU$ values of the super-sequences of $ t $ are also less than the minimum utility threshold.

\begin{theorem}
	\label{theorem:upper-bound-meu}
	Given a quantitative sequential database $ D $ and a sequence $ t $, we can obtain that 
	\begin{equation}
	u(t)\leq PEU(t)
	\end{equation}
\end{theorem}

\begin{proof}
	Since $ u(t, s) $ = $ max\{u(s_{k})|t \sim s_{k} \wedge s_{k} \subseteq s\} $ $\leq  max\{u(s_{k})+u($$<$$s - s_{k}$$>$$_{rest})|t \sim s_{k} \wedge s_{k} \subseteq s\}$ = $PEU(t, s)$. Thus, $ u(t) $ = $\displaystyle \sum_{s\in D}u(t, s)\leq \sum_{s\in D}\{PEU(t, s)|t\subseteq s\}$ = $PEU(t)$.
\end{proof}

Theorems \ref{theorem:meu} and \ref{theorem:upper-bound-meu} ensure that the complete set of HUSPs can be discovered. If the $PEU$ of a sequence $ t $ is less than the minimum utility threshold, then the utility of $ t $ is less than the minimum utility threshold, and the utilities of the super-sequences of $ t $ are also less than the minimum utility threshold.

\begin{theorem}
	\label{three-upper-bounds}
	For any quantitative sequential database \textit{D} and a sequence \textit{t}, the following relationship holds: 
	\begin{equation}
		PEU(t)\leq SEU(t)\leq SWU(t)
	\end{equation}
\end{theorem}

\begin{proof}
	Since $ u(s_{k}) \leq max\{u(s_{k}) | t \sim s_{k} \wedge s_{k} \subseteq s\}$ = $u(t, s) $ and $ u($$<$$s - s_{k}$$>$$_{rest}) \leq u($$<$$s - t$$>$$_{rest}) $, $ PEU(t ,s)$ = $max\{u(s_{k}) + u($$<$$s - s_{k}$$>$$_{rest}) | t \sim s_{k} \wedge s_{k} \subseteq s\} \leq u(t,s) + u($$<$$s - t$$>$$_{rest}) \leq u(s) $. Thus $\displaystyle PEU(t) $ = $\sum_{s\in D} \{PEU(t,s) | t \subseteq s\} \leq \sum_{s\in D} \{u(t,s) + u($$<$$s - t$$>$$_{rest}) | t \subseteq s\}$ = $SEU(t) $ $	\leq \sum_{s\in D} \{u(s) | t \subseteq s\}$ = $SWU(t)$.
\end{proof}

Theorem \ref{three-upper-bounds} indicates that the $PEU$ model is a tighter upper-bound compared to the \textit{SEU} and $SWU$ upper-bounds. Based on the $PEU$ model, the designed algorithm can prune more candidates than using the \textit{SEU} and $SWU$ models. The $PEU$ model can be used to estimate the utility values of  candidate sequences and their super-sequences. Based on the definition of HUSP and the above theorems, it can be found that if the $PEU$ of a sequence $ t $ is less than the minimum utility threshold, $ t $ and the super-sequences of $ t $ are not HUSPs.  Thus, the candidate sequences having $PEU$ values that are less than the minimum utility threshold are discarded from the candidate set by the proposed algorithm so that their child nodes (super-sequences) are not generated and explored in the LS-tree.

\subsection{Pruning Strategies}
A large amount of candidates may be generated from a candidate sequence $ t $ by performing \textit{I-Concatenations} and \textit{S-Concatenations} with items. To reduce the number of candidate sequences, this paper proposes a look ahead strategy (LAS) to eliminate unpromising candidate items early.

\begin{theorem}
	\label{theorem:las}
	Given a sequence $ t $ and a quantitative sequential database $ D $, two situations are considered to generate a super-sequence:
	
	1)	if $ i_{j} $ is a \textit{I-Concatenation} candidate item of $ t $, the maximal utility of $<$$t \oplus i_{j}$$>$$_{I-Concatenation}$ is no more than $\displaystyle \sum_{s\in D} \{PEU(t,s)|$$<$$t \oplus i_{j}$$>$$_{I-Concatenation} \subseteq s\} $.
	
	2)	if $ i_{j} $ is a \textit{S-Concatenation} candidate item of $ t $, the maximal utility of $<$$t$$ \oplus i_{j}$$>$$_{S-Concatenation}$ is no more than $\displaystyle \sum_{s\in D} \{PEU(t,s)|$$<$$t \oplus i_{j}$$>$$_{S-Concatenation} \subseteq s \} $.
\end{theorem}

\begin{proof}
	For 1), let $ t' $ = $<$$t\oplus i_{j}$$>$$_{I-Concatenation}$ for convenience. By  Theorem \ref{theorem:meu}, $ PEU(t', s)\leq PEU(t, s) $. Based on Theorem \ref{theorem:upper-bound-meu}, $ u(t') \leq PEU(t')$.	Thus $\displaystyle u(t') \leq PEU(t') $ = $ \sum_{s\in D} \{PEU(t', s) | t' \subseteq s \} \leq \sum_{s\in D} \{PEU(t, s) | t' \subseteq s \} $. 
	In the same way, 2) holds.
\end{proof}

\textbf{Look Ahead Strategy (LAS):}
Given a sequence $ t $ and a quantitative sequential database $ D $, two situations are considered: 

1). If $ i_{j} $ is a \textit{I-Concatenation} candidate item for $ t $ and $\displaystyle \sum_{s\in D} \{PEU(t,s) | $$<$$t \oplus i_{j}$$>$$_{I-Concatenation} \subseteq s\} $ is less than the minimum utility threshold, $ i_{j} $ should be removed from $ C^{I} $ (the set of candidate items for \textit{I-Concatenation} with $ t $);

2). If $ i_{j} $ is a \textit{S-Concatenation} candidate item for $ t $ and  $\displaystyle \sum_{s\in D} \{PEU(t,s) | $$<$$t \oplus i_{j}$$>$$_{S-Concatenation} \subseteq s \} $ is less than the minimum utility threshold, $ i_{j} $ should be removed from $ C^{S} $ (the set of candidate items for \textit{S-Concatenation} with $ t $).

The LAS strategy can be used to quickly remove unpromising candidate items so that they are not considered for  \textit{I-Concatenation} and \textit{S-Concatenation} of a  sequence $ t $. Thus, this strategy is useful to avoid calculating the $PEU$ values of $<$$t\oplus i_{j}$$>$$_{I-Concatenation}$ and $<$$t$$\oplus i_{j}$$>$$_{S-Concatenation}$ for each removed item $i_j$. 
Since the upper-bound can be calculated from the utility linked lists of $ t $, LAS can remove unpromising candidate items in advance. As a result, the execution time of the algorithm can be reduced since a smaller set of candidate items are considered for concatenations with $ t $.

The downward closure property based on the PEU model provides a tight upper-bound to reduce the search space for mining HUSPs. However, several useless items appear in the extensions of sequences in each transaction, which may lead to high upper-bound values. To further reduce the search space, an irrelevant item pruning strategy (IPS) is designed as follows.

\begin{theorem}
	For any sequence $ t $ and  item $ i_{j}\in I(t)_{rest} $, the maximal utility of $ <$$t \oplus i_{j}$$>$$_{I-Concatenation} $ or $<$$t \oplus i_{j}$$>$$_{S-Concatenation} $ is no more than $\displaystyle \sum_{s\in D} \{PEU(t, s) | ($$<$$t \oplus i_{j}$$>$$_{I-Concatenation} \subseteq s) $ $ \vee $ $($$<$$t \oplus i_{j}$$>$$_{S-Concatenation} \subseteq s) \} $. 
\end{theorem}

\begin{proof}
	For a concatenation that  $ <$$t \oplus i_{j}$$>$$_{I-Concatenation} $, based on Theorem \ref{theorem:las}, we have $\displaystyle \sum_{s\in D} \{PEU(t, s) $$ | ($$<$$t \oplus i_{j}$$>$$_{I-Concatenation} \subseteq s)$ $ \vee$ $ ($$<$$t \oplus i_{j}$$>$$_{S-Concatenation} \subseteq s) \} \geq \sum_{s\in D} \{PEU(t,s) | $$<$$t \oplus i_{j}$$>$$_{I-Concatenation} \subseteq s \} \geq u($$<$$t \oplus i_{j}$$>$$_{I-Concatenation}) $. 
	A similar proof can be done for $ <$$t \oplus i_{j}$$>$$_{S-Concatenation} $. 
\end{proof}

\textbf{Irrelevant Item Pruning Strategy (IPS):} 
Given a sequence $ t $ and an item $ i_{j} \in I(t)_{rest} $, if $\displaystyle \sum_{s\in D} \{PEU(t, s) | ($$<$$t \oplus i_{j}$$>$$_{I-Concatenation} \subseteq s) $ $\vee$ $ ($$<$$t \oplus i_{j}$$>$$_{S-Concatenation} \subseteq s) \} $ is less than the minimum utility threshold, $ i_{j} $ is called an irrelevant item of $ t $ and  should be removed from the utility linked lists of $ t $ and $ t $'s supersets.

With the help of the IPS, the remaining utility values of candidate sequences in each transaction decrease, since many irrelevant items can be ignored. As a result, the $PEU$ values of candidate sequences can  greatly decrease, and more candidates may be removed by the IPS.

Using the LAS and IPS pruning strategies, the designed algorithm can eliminate a large number of candidates. 
Consider a sequence $t$ that is processed by the algorithm.
First, the candidate items for \textit{I-Concatenation} and \textit{S-Concatenation} with $ t $ are pruned by the IPS,  and the UL-lists of $ t $ are recalculated. Then, the candidate items for \textit{I-Concatenation} and \textit{S-Concatenation} of the processed sequence $ t $ are assessed using the LAS instead of their $SWU$ values. 

Then, the designed algorithm generates new sequences by concatenating the processed sequence $ t $ with the candidate items. If the utility of the newly explored candidate sequence is no less than the minimum utility threshold, it is a HUSP. 
By applying the downward closure property, the  $PEU$ of the new sequence is  then checked to decide whether its super-sequences should be explored. 

\subsection{The HUSP-ULL Algorithm}

Based on the designed utility-linked (UL)-list structure (Section 4.2), the downward closure property (Section 4.3), and the above pruning strategies (Section 4.4), the designed algorithm named HUSP-ULL (High-Utility Sequential Pattern mining with UL-list) is proposed in this section.

The pseudo-code of the HUSP-ULL algorithm  is given in Algorithm~\ref{label1}. It first scans the quantitative sequential database $ D $ to calculate $ u(s) $ and build the UL-list of each $q$-sequence $ s\in D $ (Line 1). For each item $ i_{j} \in D $, the algorithm builds the projected database \textit{PD}($<$$i_{j}$$>$) to store the UL-lists of the transformed transactions (Line 4). 
The utility and $SWU$ of each 1-sequence are calculated using the corresponding projected database (Line 5). The 1-sequences having $SWU$ values that are no less than the minimum utility threshold are considered as candidate HUSPs (Lines 5 to 10). Thus, those 1-sequences with low $SWU$ values that are exactly deemed unpromising for \textit{I}-\textit{Concatenation} or \textit{S}-\textit{Concatenation}, they are moved in this step. And the 1-sequences having utilities that are no less than the minimum utility threshold are output as HUSPs (Lines 6 to 8). Using the special set of candidate HUSPs that were eliminated before, the HUSP-ULL algorithm can begin the projection growth with the built projected database \textit{PD}($<$$i_{j}$$>$). Next, the candidate HUSPs are considered as \textit{prefix} by the \textbf{PGrowth} procedure for mining larger HUSPs (Line 12). 


\begin{algorithm}
	\caption{HUSP-ULL}
	\label{label1}
	\begin{algorithmic}[1]
		\REQUIRE {\textit{D}, a quantitative sequential database; \textit{utable}, a utility table containing the unit profit of each item; $\delta$, the minimum utility threshold.} 
		\ENSURE {The set of \textit{HUSPs}.}
		\STATE scan $ D $ to: 1). calculate $u(s)$ for each $s\in D $ and calculate $ u(D) $; 2). build the UL-list of each $s\in D$;
		\STATE $HUSPs \gets \varnothing $;
		\FOR {each $ i_{j}\in D $}
		\STATE \textit{PD}($ < $$ i_{j} $$ > $)$\gets$\{the UL-list of $s|$$<$$i_{j}$$>$$\subseteq s\wedge s\in D\}$;
		\STATE calculate $SWU$($<$$i_{j}$$>$) and $u$($<$$i_{j}$$>$);
		\IF {$SWU$($<$$i_{j}$$>$)$ \geq \delta \times u(D) $}
		\IF {$u$($<$$i_{j}$$>$)$\geq \delta \times u(D)$)}
		\STATE $HUSPs$$\gets$$HUSPs$$\cup$$<$$i_{j}$$>$;
		\ENDIF
		\STATE \textbf{PGrowth}($ < $$ i_{j} $$ > $, \textit{PD}($ < $$ i_{j} $$ > $), \textit{HUSPs});
		\ENDIF
		\ENDFOR
		\STATE \textbf{return} \textit{HUSPs}
	\end{algorithmic}
\end{algorithm}


The \textbf{PGrowth} procedure (Algorithm \ref{pgrowth}) performs a depth-first search  to enumerate sequences by following the \textit{sequence-ascending} order. Sequences are enumerated by applying the \textit{I-Concatenation} and \textit{S-Concatenation} operations. 
The algorithm first removes irrelevant items and then recalculates the UL-list, as per the proposed IPS pruning strategy (Line 1).
Then, the algorithm scans the reduced projected database \textit{PD}(\textit{prefix}) to obtain $ C^{I} $ (the set of candidate items to be used for \textit{I-Concatenation}) (Line 2). 
To reduce the number of candidate items for \textit{I-Concatenation} with the sequence \textit{prefix},  the upper-bound values of the candidate items are calculated using  \textit{PD}(\textit{prefix}). 
Based on the proposed LAS pruning strategy, a candidate item $ i_{j} $ is discarded if its upper-bound is less than the minimum utility threshold. The reduced set of candidate items for \textit{S-Concatenation} with the sequence \textit{prefix}, denoted as $C^{S}$, is obtained in the same way (Line 6). 
After the concatenation operations are performed, the newly generated sequences are evaluated by applying the \textbf{Judge} procedure (Lines 4 and 8), which is explained next.

\begin{algorithm}
	\caption{\textbf{PGrowth}(\textit{prefix}, \textit{PD}(\textit{prefix}), $ HUSPs $)}
	\label{pgrowth}
	\begin{algorithmic}[1]
		\STATE $ PD(prefix) \gets $ IPS($ PD(prefix) $)  \COMMENT{measured by IPS}
		\STATE scan $ PD(prefix) $ to get $ C^{I} $;  \COMMENT{measured by LAS}
		\FOR {each $ i_{j}\in C^{I} $} 
		\STATE \textbf{Judge}($<$\textit{prefix}$\oplus i_{j}$$>_{I-Concatenation}$, \textit{PD}(\textit{prefix}), \textit{HUSPs});
		\ENDFOR
		\STATE scan $ PD(prefix) $ to get $ C^{S} $;  \COMMENT{evaluated using LAS}
		\FOR {each $ i_{j}\in C^{S} $}
		\STATE \textbf{Judge}($ < $\textit{prefix}$\oplus$$i_{j}$$>_{S-Concatenation}$, \textit{PD}(\textit{prefix}), \textit{HUSPs});
		\ENDFOR
	\end{algorithmic}
\end{algorithm}

The \textbf{Judge} procedure (Algorithm \ref{judge}) first builds the projected database \textit{PD}(\textit{prefix}')  from \textit{PD}(\textit{prefix}) (Line 1). 
The $PEU$ and utility of \textit{prefix'} are then calculated from \textit{PD}(\textit{prefix'}) (Line 2). If the utility of \textit{prefix'} is no less than the minimum utility threshold, \textit{prefix'} is identified as a HUSP (Lines 4 to 5). If the $PEU$ of \textit{prefix'} is no less than the minimum utility threshold, the \textbf{PGrowth} procedure is then applied with \textit{prefix'} to discover HUSPs by considering the super-sequences of \textit{prefix'} (Lines 3 to 6). The algorithm terminates if no candidates are generated. Finally, the designed algorithm  returns the set of discovered HUSPs.

\begin{algorithm}
	\caption{\textbf{Judge}(\textit{prefix'}, \textit{PD}(\textit{prefix}), \textit{HUSPs})}
	\label{judge}
	\begin{algorithmic}[1]
		\STATE \textit{PD}(\textit{prefix'})$\gets$\{UL-list of $s|$\textit{prefix'}$\subseteq s\wedge s\in$ \textit{PD}(\textit{prefix})\};
		\STATE calculate $ u $(\textit{prefix'}) and $ PEU $(\textit{prefix'});
		\IF {$ PEU $(\textit{prefix'})$\geq \delta \times u(D)$}
		\IF {$ u $(\textit{prefix'})$\geq \delta \times u(D) $}
		\STATE $ HUSPs $$\gets$$ HUSPs $$\cup$\textit{prefix'};
		\ENDIF
		\STATE \textbf{PGrowth}(\textit{prefix'}, \textit{PD}(\textit{prefix'}), \textit{HUSPs});
		\ENDIF
	\end{algorithmic}
\end{algorithm}

\section{Experimental Results}
\label{sec:experiment}

In this section, we conduct experiments on several real datasets to showcase the advantage of HUSP-ULL in the task of high-utility sequential pattern mining. In particular, we aim to answer the following research questions via the experiments:

$ \bullet $ How effectively HUSP-ULL can discover the useful high-utility sequential patterns with observed timestamps from the quantitative sequential datasets?

$ \bullet $ How HUSP-ULL benefits from each component of the propose structure and the developed pruning strategies for mining HUSPs?

$ \bullet $ How efficiently HUSP-ULL can be applied when handling large data with different sizes?

\begin{table}[!htbp] 	
	\setlength{\abovecaptionskip}{0pt}
	\setlength{\belowcaptionskip}{0pt} 
	\caption{Parameters of the Datasets.}
	\centering
	\begin{tabular}{|c|c|}
		\hline
		$\mathbf{\#|D|}$ &	Number of sequences \\ \hline
		$\mathbf{\#|I|}$ &	Number of distinct items     \\ \hline
		\textbf{C} & Average number of itemsets per sequence   \\ \hline
		\textbf{T} & Average number of items per itemset   \\ \hline
		\textbf{MaxLen} &	Maximum number of items per sequence                \\ \hline 
	\end{tabular}
	\label{table:paras}
\end{table}

\begin{table}[!htbp] 	
	\setlength{\abovecaptionskip}{0pt}
	\setlength{\belowcaptionskip}{0pt} 
	\caption{Characteristics of the Datasets.}
	\centering
	\begin{tabular}{|c|c|c|c|c|c|}
		\hline
		\textbf{Dataset}		& $\mathbf{\#|D|}$ & $\mathbf{\#|I|}$ &	\textbf{C} & \textbf{T} & \textbf{MaxLen}  \\ \hline 
		Sign                    & 730           &   267     & 52.0  & 1 & 94 \\ \hline
		Bible                   & 36,369        & 13,905    & 21.6  & 1 &   100 \\ \hline
		SynDataset-160k         & 159,501        & 7,609       & 6.19   & 4.32 & 20 \\ \hline
		Kosarak10k              & 10,000        & 10,094    & 8.14   & 1 &   608 \\ \hline
		Leviathan               & 5,834          & 9,025      & 33.8  & 1 & 100 \\ \hline
		yoochoose-buys          & 234,300         & 16,004     & 1.13  & 1.97 & 21 \\ \hline
	\end{tabular}
	\label{table:chars}
\end{table}

\begin{figure*}[!htbp]
	\setlength{\abovecaptionskip}{0pt}
	\setlength{\belowcaptionskip}{0pt}	
	\centering
	\includegraphics[trim=20 0 20 0,clip,scale=0.56]{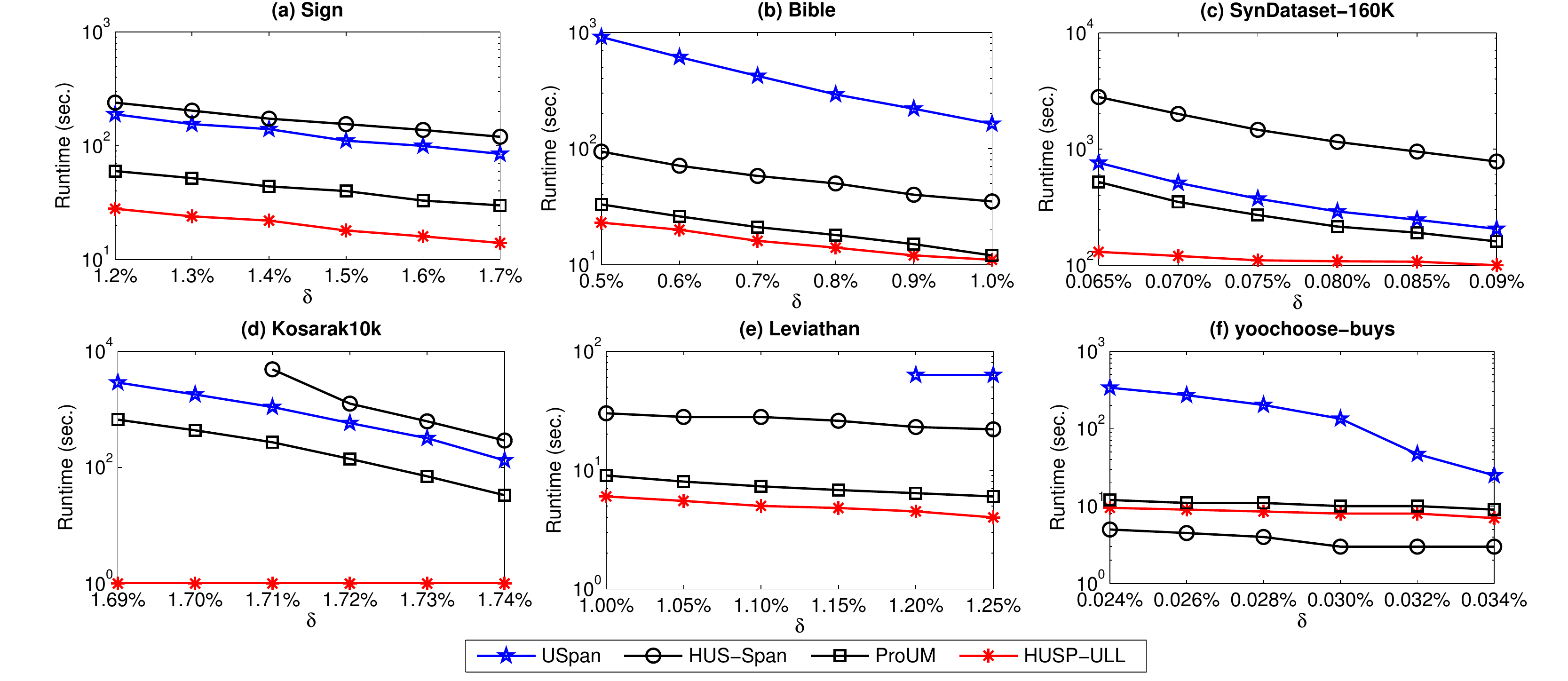}
	\caption{Runtime for various $ \delta $ values.}
	\label{fig:runtime1}
\end{figure*}

\subsection{Datasets}
Totally six real-life datasets \cite{fournier2016spmf} and one synthetic dataset were used in the experiments to evaluate the performance of the proposed algorithm. 

$ \bullet $ \textbf{Sign}  is a real-life dataset of sequences of sign language utterance, created by the National Center for Sign Language and Gesture Resources at Boston University. Each utterance in the dataset is associated with a segment of video with a detailed transcription.
	
$ \bullet $ \textbf{Bible} is a real-life dataset obtained by converting the Bible into a set of sequences of items (words).

$ \bullet $ \textbf{SynDataset-160K} is a synthetic dataset that generated by IBM Quest Dataset Generator \cite{agrawal1994dataset}. It contains 160,000 sequences. This synthetic sequential dataset with different data sizes (from 10,000 sequences to 400,000 sequences, named  C8S6T4I3D$ | $X$ | $K) were also used to evaluate the scalability of the compared approaches. 
	
$ \bullet $ \textbf{Kosarak10k} is a real-life dataset of click-stream data from a Hungarian news portal, which is a subset of the original Kosarak dataset \cite{FIM-dataset}. 

$ \bullet $ \textbf{Leviathan} is a conversion of Thomas Hobbes' Leviathan novel (1651) to a sequence of items (words). 
	
$ \bullet $ \textbf{yoochoose-buys} commercial dataset was constructed by YOOCHOOSE GmbH to support participants in the RecSys Challenge 2015\footnote{\url{https://recsys.acm.org/recsys15/challenge/}}. It contain a collection of 1,150,753 sessions from a retailer, where each session is encapsulating the click events. The total number of item IDs and category IDs is 54,287 and 347 correspondingly, with an interval of 6 months.

Parameters and characteristics of these datasets are respectively shown in Table \ref{table:paras} and Table \ref{table:chars}. In the field of utility mining, a simulation model \cite{liu2005two} was widely used in the previous studies \cite{tseng2013efficient,zida2015efim,wang2016efficiently} to generate the quantities and unit profit values of items in the sequential datasets.  In order to achieve a fair comparison, this simulation model \cite{liu2005two} was adopted in our experiments. Note that the quantity of each item is randomly generated in the [1, 5] interval. A log-normal distribution was used to randomly assign  profit values of items in the [0.01, 10.00] interval.  The above datasets can be downloaded from \cite{fournier2016spmf}.

\subsection{Experimental Settings}

All the compared algorithms were implemented in Java. The experiments were carried out on a personal computer equipped with an Intel(R) Core(TM) i7-7700HQ CPU @ 2.80 GHz 2.81 GHz, 32 GB of RAM, running the 64-bit Microsoft Windows 10 operating system.  

Although many utility-oriented HUSPM algorithms have been developed before, we conduct experiments against the following state-of-the-art HUSPM methods:

$ \bullet $ \textbf{USpan} \cite{yin2012uspan}: It is one of the commonly compared HUSP algorithm that uses utility matrix and two upper-bounds on utility for width and depth pruning. In the experiments, the USpan algorithm was fixed and replaced its upper bound by \textit{SEU} \cite{gan2019proum}. 

$ \bullet $ \textbf{HUS-Span} \cite{wang2016efficiently}: It is a SWU-based framework designed for identifying high-utility sequences. This method combines  two quantitative metrics, called Reduced Sequence Utility (\textit{RSU}) and Projected-Database Utility (\textit{PDU}), to prune low utility sequences. 

$ \bullet $ \textbf{ProUM} \cite{gan2019proum}: This projection-based model utilizes the sequence extension utility (\textit{SEU}) to present the maximum utility  of  the  possible  extensions  that  based  on  the  prefix. Besides, it applies the project mechanism during the construction of the utility-array, which makes it achieves the better performance than previous HUSPM algorithms, e.g., USpan, HUS-Span.

In the following sections, substantial experiments were performed to evaluate the effectiveness and efficiency of the proposed HUSP-ULL algorithm. Notice that the source code and test datasets will be released at the well-known SPMF \cite{fournier2016spmf} \footnote{\url{http://www.philippe-fournier-viger.com/spmf/}} data mining platform after the acceptance for publication.

\begin{figure*}[!htbp]
	\setlength{\abovecaptionskip}{0pt}
	\setlength{\belowcaptionskip}{0pt}	
	\centering
	\includegraphics[trim=10 0 10 0,clip,scale=0.54]{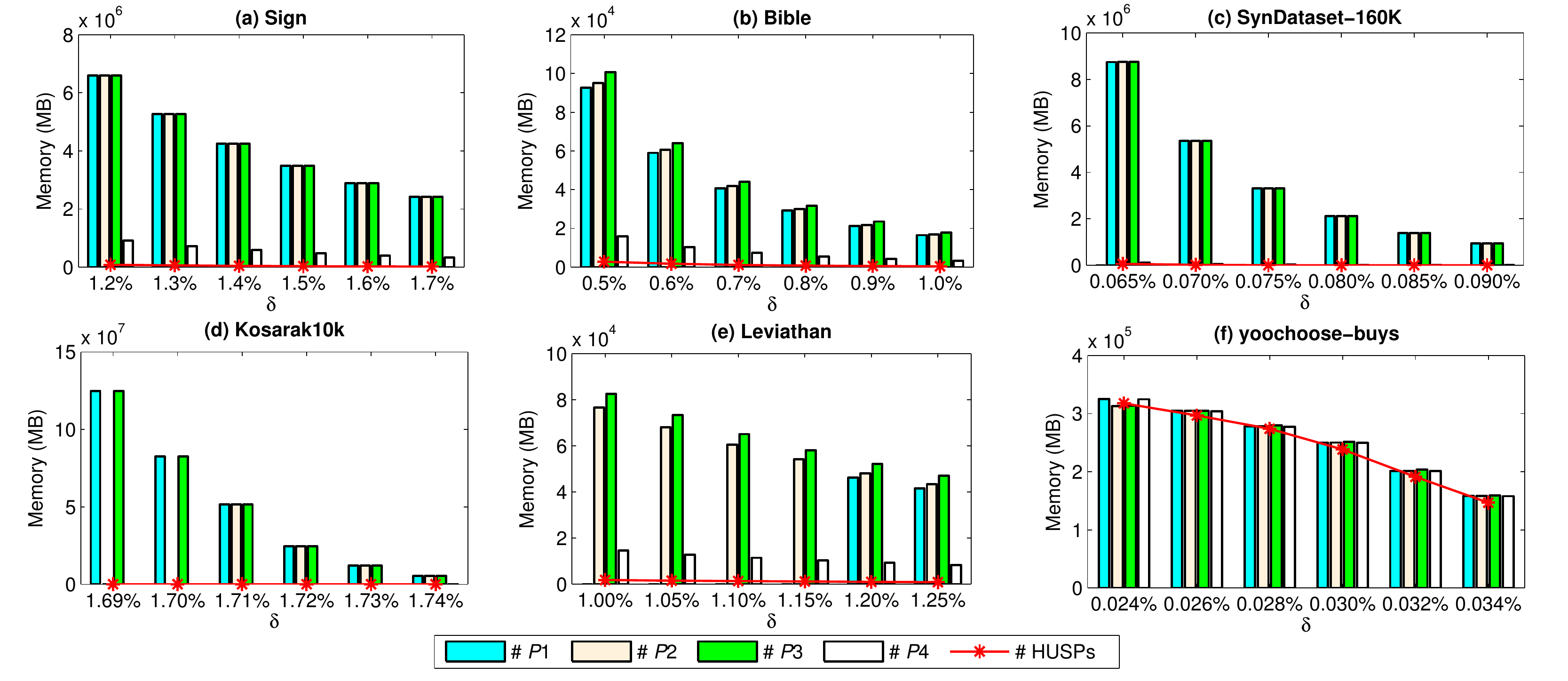}
	\caption{Number of patterns (candidates and final results) under various $ \delta $ values.}
	\label{fig:pattern}
\end{figure*}

\subsection{Efficiency}
In a first experiment, the runtime of the designed algorithm was compared with that of the existing state-of-the-art algorithms. The runtime was measured by considering both the time used by the CPU and the time required for disk  I/O accesses. Fig. \ref{fig:runtime1} shows the runtime of the compared algorithms for various minimum utility thresholds (denoted as $ \delta $ values).

We now discuss the result concerning the efficiency of HUSP-ULL. It can be seen in Fig. \ref{fig:runtime1} that the proposed HUSP-ULL algorithm outperforms other approaches for all datasets except for yoochoose-buys under various threshold values. Generally, the proposed HUSP-ULL  algorithm is faster than the other algorithms by at least one order of magnitude. For example, in Figs. \ref{fig:runtime1} (c), (d), (e), the  HUS-Span  and USpan algorithms spend more than 1000 seconds, and in some cases, cannot even terminate in a reasonable time. In contrast, HUSP-ULL  spends less than 100 seconds to output the results  for these threshold values.

As $ \delta $ is decreased, the compared approaches become slower. The runtime of ProUM and HUSP-ULL increases smoothly, while the runtime of the compared USpan and HUS-Span algorithms increases more rapidly. For example,  in the case we can see in Fig. \ref{fig:runtime1} (c) that the runtime of HUS-Span and USpan increases dramatically while the threshold values are only slightly changed. Thus, the runtime performance of HUS-Span and USpan are very sensitive with respect to the parameter settings. Generally, when $ \delta $ is set  to a small value, the runtime of HUS-Span and USpan sharply increases due to their actual search space  and the large number of candidates that they generated. Thus, it demonstrates that the designed utility-linked list-based HUSP-ULL algorithm are able to significantly improve the performance in terms of running time.

It is important to note that the USpan algorithm outperforms the HUS-Span algorithm in Figs. \ref{fig:runtime1} (a), (c), (d), while HUS-Span  outperforms USpan in Fig. \ref{fig:runtime1} (b), (e), (f).  In most cases, the projection-based ProUM algorithm performs better than USpan and HUS-Span. Besides, it can be seen that the USpan algorithm cannot return results in Fig. \ref{fig:runtime1} (e) since it runs out of memory. The USpan algorithm builds a series of utility-matrix to store utility information about patterns, but it requires additional processing time. Thus the runtime of USpan  is larger than that of ProUM in many cases.

\subsection{Effectiveness of Pruning Strategies}

In order to evaluate the effectiveness of pruning strategies, the number of generated candidates of all compared algorithms and the number of discovered high-utility sequential patterns (\textbf{\# HUSPs}) under different parameter settings are compared in this section. The results are shown  in Fig. \ref{fig:pattern}. Note that \textbf{\# \textit{P1}}, \textbf{\# \textit{P2}}, \textbf{\# \textit{P3}}, and \textbf{\# \textit{P4}}  denote the number of the candidate patterns generated by USpan, HUS-Span, ProUM, and HUSP-ULL, respectively. And \textbf{\# \textit{HUSPs}} denote the number of final HUSPs discovered by the three compared algorithms. Firstly, noticing that the searching task has its running time exceeds 10,000 seconds or out of memory (a maximum of 4096 MB (4 GB) of memory setting) when searching candidates and HUSPs, as shown with the notation ``-".

It can be seen in Fig.  \ref{fig:pattern} that the number of candidates generated by the HUSP-ULL algorithm is much less than that of the other algorithms. This shows that the designed HUSP-ULL algorithm and pruning strategies can greatly reduce the number of unpromising candidates for mining the HUSPs and hence reduces the requirements in terms of runtime and memory. In all test datasets, the number \textbf{\# \textit{P3}} is close to the number of \textbf{\# \textit{P2}}. It indicates that the upper bound named \textit{SEU} used in ProUM has the similar overestimate effect when compared to the \textit{PEU} upper bound used in HUS-Span.

As the minimum utility threshold $ \delta $ is decreased, the number of candidates increases for the  HUS-Span, USpan, and ProUM algorithms. In contrast, that number increases much more slowly for the HUSP-ULL algorithm. When the minimum utility threshold $ \delta $ is set lower, it is obvious that the HUSP-ULL algorithm generates much fewer candidates than the other algorithms. Especially, the USpan algorithm generates no results in Fig. \ref{fig:pattern} (d) due to a very large number of candidates.

\begin{figure*}[!htbp]
	\setlength{\abovecaptionskip}{0pt}
	\setlength{\belowcaptionskip}{0pt}	
	\centering
	\includegraphics[trim=0 0 10 0,clip,scale=0.54]{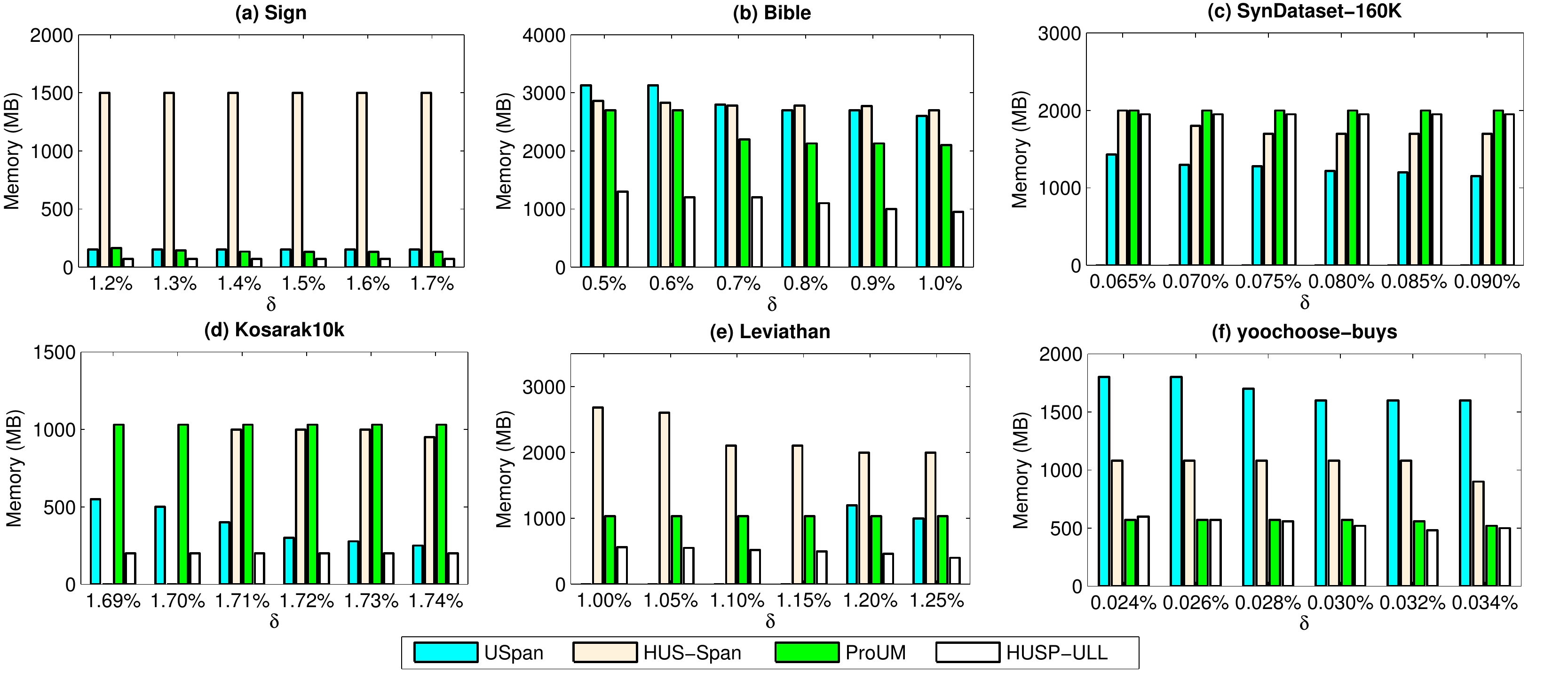}
	\caption{Memory usage under various $ \delta $ values.}
	\label{fig:memory}
\end{figure*}

\begin{figure*}[!htbp]
	\setlength{\abovecaptionskip}{0pt}
	\setlength{\belowcaptionskip}{0pt}	
	\centering
	\includegraphics[trim=20 200 30 0,clip,scale=0.55]{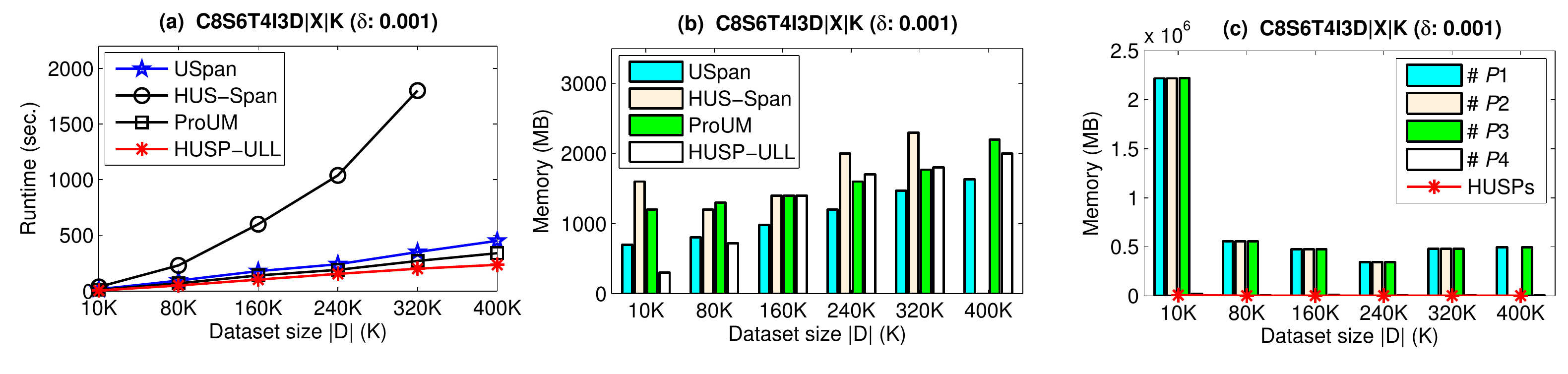}
	\caption{Scalability of the compared approaches.}
	\label{fig:scalability}
\end{figure*}

It can also be observed that the number of candidates generated by the USpan algorithm (with the \textit{SEU} upper bound) is close to that of the HUS-Span algorithm in most cases. For all compared HUSPM algorithms, the number of the final results of HUSPs is quite less than that of the generated candidates, such as \textbf{\# HUSPs} is less than \textbf{\# \textit{P1}}, \textbf{\# \textit{P2}}, \textbf{\# \textit{P3}}, and \textbf{\# \textit{P4}},  as shown in  Fig. \ref{fig:pattern} (a) to (e).  Although ProUM  uses a structure named utility-array to store sequences and utility information in memory, it still generates a huge number of candidate patterns for discovering high-utility sequential patterns.  The proposed HUSP-ULL algorithm employs the UL-list structure to speed up the mining process and uses projected databases to reduce memory consumption.

Even though the LS-tree may theoretically grow very large, in practice it stays relatively small in the proposed HUSP-ULL framework. We only consider a small part of the candidate space. That is, we only perform the \textit{I-Concatenation} and \textit{S-Concatenation} by combining the potential candidate patterns that may be the  promising high-utility patterns. In practice, as shown in the Bible dataset, we can find that HUSP-ULL only has to cache up to a few thousand candidates. Using the proposed pruning strategies in LS-tree, HUSP-ULL can speed up in computation, up to an order of magnitude, while the memory consumption is also reduced. 

\subsection{Memory Usage Evaluation}

For the applications of data mining and analytics, the memory usage of a data mining algorithm is one of the key measure criteria. Therefore, to show a good efficiency, it would be better to test memory usage in performance evaluation. In this subsection, we further evaluate the memory usage of all the compared algorithms. With the same parameter setting as run in Fig. \ref{fig:runtime1}, the memory usage of each algorithm under various $ \delta $ values are shown in Fig. \ref{fig:memory}.

As mentioned early, the maximum memory is set to 4096 MB, and USpan is run out of memory in Leviathan. It is clear that the proposed HUSP-ULL algorithm consumes the least memory among the compared algorithms with all parameter settings on all datasets, except for the SynDataset-160K. Among the compared algorithms, the memory usage of HUSP-ULL is always very stable. For example, under six varied $\delta$, it consumes around 200 MB on Kosarak10k, and consumes from 600 MB to 500 MB in yoochoose-buys. However, it can be observed that there is a sharp decrease in USpan and HUS-Span on some cases, as shown in Kosarak10k, Leviathan and yoochoose-buys. For example, USpan was run out of memory when $\delta$ is set less than 1.20\% in Leviathan.

It is also interesting to observed that HUS-Span sometimes may consume more memory than the utility matrix based USpan algorithm, as shown in Sign dataset. To summarize, in most cases on the test datasets, the proposed HUSP-ULL algorithm  significantly outperforms the state-of-the-art HUSPM algorithms in terms of memory consumption. The reason is that HUSP-ULL utilizes the compact UL-list and two pruning strategies to reduce the space complexity.

\subsection{Scalability}

We further evaluated the scalability of the compared approaches on the synthetic dataset C8S6T4I3D$ | $X$ | $K \cite{agrawal1994dataset} (recall that $ | $D$ | $ is the size of the dataset $D$). The results in terms of runtime and  number of candidates for  different threshold values are shown in Fig. \ref{fig:scalability}. The size of the synthetic dataset is varied from 10K to 400K sequences, with a threshold $\delta$: 0.001 at each test.

In Fig. \ref{fig:scalability}, it can be observed that the HUSP-ULL algorithm has better scalability than the compared state-of-the-art algorithms in large dataset. As the dataset size is increased, the runtime of  HUS-Span, USpan, ProUM, and HUSP-ULL increases, respectively. Notice that the HUS-Span algorithm does not return any results in some cases because it runs out of memory when the minimum utility threshold is set to a small value or when the dataset is very large. This is because the HUS-Span algorithm utilizes the utility-chain structure to store  utility information, which requires a large amount of memory, to speed up the mining process. If patterns match many  transactions, this structure can consume a large amount of memory. 

In Fig. \ref{fig:scalability} (c), it  can be found that the number of candidates does not increase when the dataset size is increased. This is reasonable since the minimum utility value (w.r.t. the value of $ \delta $ $\times $ $u(D) $)  increases as the dataset size is increased. Hence, fewer candidates are HUSPs, but the algorithms still spend time  to evaluate  candidates. Thus, the runtime increases with the dataset size.  
\section{Conclusion and Future Work}
\label{sec:conclusion} 
Utility-based sequence mining is a significant problem due to the subtle interesting patterns among different factors (e.g., timestamp, quantity, profit) and the meaningful knowledge triggered by complex real-life situations. This paper has proposed  a novel HUSP-ULL algorithm to discover high-utility sequential patterns (HUSPs) more efficiently. Specifically, the concept of utility-linked (UL)-list was developed and used to calculate the utilities and the upper-bound values of candidates for deriving all HUSPs. By utilizing the designed LS-tree, UL-list structure, the HUSP-ULL algorithm can fast discover the complete set of HUSPs. To further improve the performance of the proposed HUSP-ULL algorithm, two pruning strategies were introduced to reduce the upper-bounds on utility and thus prune the search space to find  HUSPs. Substantial experiments on some real datasets show that the designed algorithm can effectively and efficiently identify all HUSPs and outperforms the state-of-the-art HUSPM algorithms. The proposed pruning strategies also improves the effectiveness and efficiency for mining HUSPs by reducing the number of candidates. 

In the future, several extensions of the proposed algorithm can be considered such as to design more efficient algorithms, mine HUSPs in big data, or extending the model to other pattern mining problems.

\section{Acknowledgment}
We would like to thank  Dr. Jun-Zhe Wang for providing the original C++ code of the HUS-Span algorithm, and Dr.  Oznur Kirmemis Alkan for sharing the executable file of the HuspExt algorithm.  This research was partially supported by the China Scholarship Council Program.

\ifCLASSOPTIONcaptionsoff
  \newpage
\fi



\bibliographystyle{IEEEtran}
\bibliography{main}
%




\vspace{-1cm}
\begin{IEEEbiography}[{\includegraphics[width=1in,height=1.25in,clip,keepaspectratio]{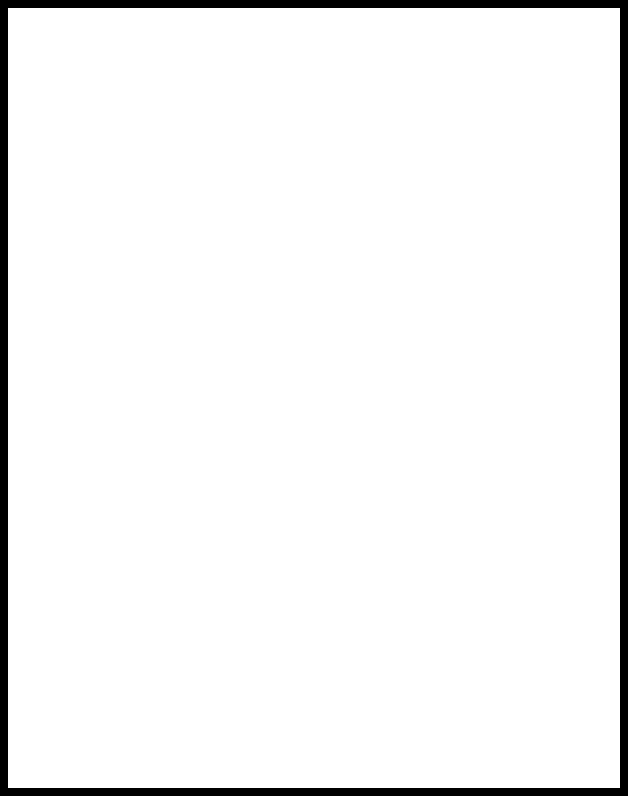}}]{Wensheng Gan} received the Ph.D. in Computer Science and Technology, Harbin Institute of Technology (Shenzhen), Guangdong, China in 2019. He received the B.S. degree in Computer Science from South China Normal University, Guangdong, China in 2013. His research interests include data mining, utility computing, and big data analytics. He has published more than 50 research papers in  peer-reviewed journals and international conferences, which have received more than 500 citations.
\end{IEEEbiography}

\vspace{-1cm}
\begin{IEEEbiography}[{\includegraphics[width=1in,height=1.25in,clip,keepaspectratio]{newAuthor.png}}]{Jerry Chun-Wei Lin (SM'19)}
	is an associate professor at Western Norway University of Applied Sciences, Bergen, Norway. He received the Ph.D. in Computer Science and Information Engineering, National Cheng Kung University, Tainan, Taiwan in 2010. His research interests include data mining, big data analytics, and social network. He has published more than 300 research papers in peer-reviewed international conferences and journals, which have received more than 3000 citations. He is the co-leader of the popular SPMF open-source data mining library and the Editor-in-Chief (EiC) of the \textit{Data Mining and Pattern Recognition} (DSPR) journal, and Associate Editor of \textit{Journal of Internet Technology}. 
\end{IEEEbiography}

\vspace{-1cm}
\begin{IEEEbiography}[{\includegraphics[width=1in,height=1.25in,clip,keepaspectratio]{newAuthor.png}}]{Jiexiong Zhang}
	is currently a senior software engineer in Didi Chuxing, Beijing, China. He received the M.S. degrees in Computer Science from Harbin Institute of Technology (Shenzhen), Guangdong, China in 2017. His research interests include data mining, artificial intelligence, and big data analytics. 
\end{IEEEbiography}

\vspace{-2cm}
\begin{IEEEbiography}[{\includegraphics[width=1in,height=1.25in,clip,keepaspectratio]{newAuthor.png}}]{Philippe Fournier-Viger}
	is full professor and Youth 1000 scholar at the Harbin Institute of Technology (Shenzhen), Shenzhen, China. He received a Ph.D. in  Computer Science at the University of Quebec in Montreal (2010). His research interests include pattern mining, sequence analysis and prediction, and social network mining. He has published more than 250 research papers in refereed international conferences and journals. He is the founder of the popular SPMF open-source data mining library, which has been cited in more than 800 research papers. He is Editor-in-Chief (EiC) of the \textit{Data Mining and Pattern Recognition} (DSPR) journal.
\end{IEEEbiography}

\vspace{-2cm}
\begin{IEEEbiography}[{\includegraphics[width=1in,height=1.25in,clip,keepaspectratio]{newAuthor.png}}]{Han-Chieh Chao (SM'04)}
	has been the president of National Dong Hwa University since February 2016. He received M.S. and Ph.D. degrees in Electrical Engineering from Purdue University in 1989 and 1993, respectively. His research interests include high-speed networks, wireless networks, IPv6-based networks, and artificial intelligence. He has published nearly 500 peer-reviewed professional research papers. He is the Editor-in-Chief (EiC) of IET Networks and \textit{Journal of Internet Technology}. Dr. Chao has served as a guest editor for ACM MONET, IEEE JSAC, \textit{IEEE Communications Magazine}, \textit{IEEE Systems Journal}, \textit{Computer Communications}, \textit{IEEE Proceedings Communications}, \textit{Wireless Personal Communications}, and \textit{Wireless Communications \& Mobile Computing}. Dr. Chao is an IEEE Senior Member and a fellow of IET. 
\end{IEEEbiography}

\vspace{-2cm}
\begin{IEEEbiography}[{\includegraphics[width=1in,height=1.25in,clip,keepaspectratio]{newAuthor.png}}]{Philip S. Yu (F'93)}
	received the B.S. degree in electrical engineering from National Taiwan University, M.S. and Ph.D. degrees in electrical engineering from Stanford University, and an MBA from New York University. He is a distinguished professor of computer science with the University of Illinois at Chicago (UIC) and also holds the Wexler Chair in Information Technology at UIC. Before joining UIC, he was with IBM, where he was manager of the Software Tools and Techniques Department at the Thomas J. Watson Research Center. His research interests include data mining, data streams, databases, and privacy. He has published more than 1,300 papers in peer-reviewed journals (i.e., TKDE, TKDD, VLDBJ, ACM TIST) and conferences (KDD, ICDE, WWW, AAAI, SIGIR, ICML, etc). He holds or has applied for more than 300 U.S. patents. Dr. Yu was the Editor-in-Chief of \textit{ACM Transactions on Knowledge Discovery from Data}. He received the ACM SIGKDD 2016 Innovation Award, and the IEEE Computer Society 2013 Technical Achievement Award. Dr. Yu is a fellow of the ACM and the IEEE.
\end{IEEEbiography}


%

%






\end{document}